\documentclass[aps,pra,10pt,superscriptaddress,twocolumn,footinbib,letterpaper,longbibliography]{revtex4-1}
\pdfoutput=1

\usepackage{natbib}
\usepackage{url}
\usepackage{textcase}
\usepackage[T1]{fontenc}
\usepackage{float}
\usepackage[final,bookmarksnumbered]{hyperref}
\usepackage[USenglish]{babel}
\usepackage{amsfonts}
\usepackage{amsthm}
\usepackage{amsmath}
\usepackage{amssymb}
\usepackage{times}
\usepackage[ascii]{inputenc}
\usepackage{aliascnt}
\usepackage{microtype}
\usepackage{graphicx}
\usepackage{mathtools}
\usepackage[all]{hypcap}
\usepackage{braket}
\usepackage{color}
\usepackage[normalem]{ulem}

\bibpunct{(}{)}{,}{n}{}{,}

\newcommand{\newjointcountertheorem}[3]{\newaliascnt{#1}{#2}\newtheorem{#1}[#1]{#3}\aliascntresetthe{#1}}
\newtheorem{thm}{Theorem}
\newjointcountertheorem{lem}{thm}{Lemma}
\newjointcountertheorem{cor}{thm}{Corollary}
\newjointcountertheorem{prp}{thm}{Proposition}
\newjointcountertheorem{cnj}{thm}{Conjecture}
\newjointcountertheorem{que}{thm}{Question}
\newjointcountertheorem{pro}{thm}{Problem}
\newjointcountertheorem{fct}{thm}{Fact}
\newjointcountertheorem{obs}{thm}{Observation}
\newjointcountertheorem{alg}{thm}{Algorithm}
\newjointcountertheorem{asm}{thm}{Assumption}
\theoremstyle{definition}
\newjointcountertheorem{dfn}{thm}{Definition}
\newjointcountertheorem{ntn}{thm}{Notation}
\theoremstyle{remark}
\newjointcountertheorem{rem}{thm}{Remark}
\newjointcountertheorem{nte}{thm}{Note}
\newjointcountertheorem{exl}{thm}{Example}

\def\Snospace~{\S{}}


\DeclareMathOperator{\grad}{grad}

\DeclareMathOperator{\diag}{diag}
\DeclareMathOperator{\Sym}{Sym}

\DeclareMathOperator{\conv}{conv}

\DeclareMathOperator{\GHZ}{GHZ}
\DeclareMathOperator{\EPR}{EPR}
\DeclareMathOperator{\SEP}{SEP}
\DeclareMathOperator{\tr}{tr}
\DeclareMathOperator{\GL}{GL}
\DeclareMathOperator{\SL}{SL}
\DeclareMathOperator{\SU}{SU}

\renewcommand{\epsilon}{\varepsilon}
\newcommand{\restrict}[2]{\left.#1\vphantom{\big|}\right|_{#2}}

\newcommand{\norm}[1]{\lVert#1\rVert}
\newcommand{\abs}[1]{\lvert#1\rvert}
\newcommand{\PP}{\mathbb P}
\newcommand{\CC}{\mathbb C}

\newcommand{\ZZ}{\mathbb Z}

\newcommand{\id}{\openone}
\newcommand{\ketbra}[2]{\lvert#1\rangle\langle#2\rvert}

\allowdisplaybreaks[4]

\begin{document}

\hypersetup{pdfauthor={Michael Walter and Brent Doran and David Gross and Matthias Christandl},pdftitle={Entanglement Polytopes: Multiparticle Entanglement from Single-Particle Information}}
%\title{Entanglement Polytopes}
%\title{Multi-Particle Entanglement from Single-Particle Information}
\title{Entanglement Polytopes: Multiparticle Entanglement from Single-Particle Information}
\author{Michael \surname{Walter}}
\affiliation{Institute for Theoretical Physics, ETH Zurich, Wolfgang--Pauli--Strasse 27, CH-8093 Zurich, Switzerland}
\author{Brent \surname{Doran}}
\affiliation{Department of Mathematics, ETH Zurich, R\"amistrasse 101, CH-8092 Zurich, Switzerland}
\author{David \surname{Gross}}
\affiliation{Institute for Physics, University of Freiburg, Rheinstrasse 10, D-79104 Freiburg, Germany}
\author{Matthias \surname{Christandl}}
\affiliation{Institute for Theoretical Physics, ETH Zurich, Wolfgang--Pauli--Strasse 27, CH-8093 Zurich, Switzerland}

\begin{abstract}
  Entangled many-body states are an essential resource for quantum computing and interferometry.
  Determining the type of entanglement present in a system usually requires access to an exponential number of parameters.
  We show that in the case of pure multi-particle quantum states, features of the global entanglement can already be extracted from local information alone.
  This is achieved by associating with any given class of entanglement an entanglement polytope---a geometric object which characterizes the
  single-particle states compatible with that class.
  Our results, applicable to systems of arbitrary size and statistics, give rise to local witnesses for global pure-state entanglement, and can be generalized to states affected by low levels of noise.
\end{abstract}

\maketitle

% brief introduction describing the paper's significance, which should
% be intelligible to readers in various disciplines.

Entanglement is a uniquely quantum mechanical feature. It is responsible for
fundamentally new effects---such as quantum non-locality---and
constitutes the basic resource for concrete tasks such as quantum
computing \cite{vidal03} and interferometry beyond the standard limit
\cite{leibfriedbarrettschaetzetal04,giovannettilloydmaccone04}.
Considerable efforts have been directed at obtaining a systematic
characterization of multi-particle entanglement; however, our understanding remains limited as the complexity of entanglement scales exponentially with the number of particles \cite{horodecki09}.

In this work, we show that, for pure quantum states, single-particle information alone can serve as a powerful witness to multi-particle entanglement.
In fact, we find that a finite list of linear inequalities characterizes the eigenvalues of the single-particle states in any given class of entanglement.
Their violation provides a criterion for witnessing multi-particle entanglement that (i) only requires access to a linear number of degrees of freedom, (ii) applies universally to quantum systems of arbitrary size and statistics, and (iii) distinguishes among many
important classes of entanglement, including genuine multi-particle entanglement.
Geometrically, these inequalities cut out a hierarchy of polytopes, which captures all information about the global pure-state entanglement deducible from local information alone. Our methods are sufficiently robust to be applicable to situations where the state is affected by low levels of noise.

Formally, a pure state is said to be entangled if it
cannot be written as
a product
% {\lvert \psi \rangle} \neq
$\lvert \psi^{(1)} \rangle \otimes \ldots \otimes \lvert \psi^{(N)} \rangle$ \cite{horodecki09}.
%As entanglement theory explores
%the non-classical aspects of quantum states,
Two states can be considered to belong to the same entanglement class if they can be converted into
each other with finite probability of success using
local operations and classical communication (stochastic
LOCC, or SLOCC)
\cite{bennettpopescurohrlichetal00,durvidalcirac00}.
For small systems, these entanglement classes are well-understood.
In the simplest scenario of three qubits (two-level systems),
there exist two classes of genuinely entangled states of very different
nature: the first contains the famous Greenberger--Horne--Zeilinger (GHZ) state
$\frac 1 {\sqrt{2}} ( {\lvert \uparrow\uparrow\uparrow \rangle} + {\lvert \downarrow\downarrow\downarrow \rangle} )$, which exhibits a particularly
strong form of quantum correlations
\cite{greenbergerhornezeilinger89}; the second contains the W state
$\frac 1 {\sqrt{3}} ( {\lvert \uparrow\uparrow\downarrow \rangle} + {\lvert \uparrow\downarrow\uparrow \rangle} + {\lvert \downarrow\uparrow\uparrow \rangle} )$
\cite{durvidalcirac00}.
Whereas states in the W class can be approximated to arbitrary precision by states from the GHZ class, the converse is not
true---implying stronger entanglement of the GHZ class \cite{durvidalcirac00}.
Already for four particles there exist infinitely many
entanglement classes \cite{verstraetedehaenedemooretal02}, and the
number of parameters required to determine the
class grows exponentially with the particle number. As a result,
only sporadic results have been obtained for larger systems, despite the enormous
amount of literature dedicated to the problem \cite{horodecki09}.

% The interest in entanglement classes stems not only
% from their role in quantum information. Indeed, connections to
% relativistic transformations %\cite{huberfriisgabrieletal11}
% and even to black holes have been found \cite{borstendahanayakeduffetal09}.
% Mathematically, the problem
% can be formulated in terms of invariant theory,
% studied since the 19th century---Cayley's hyperdeterminant, e.g.,
% appears as the 3-tangle
% \cite{coffmankunduwootters00}.
% Similar techniques underpin modern developments,
% such as the geometric complexity theory approach to the P vs.~NP
% problem \cite{burgisserlandsbergmaniveletal11}.

Our approach to multi-particle entanglement is based on establishing a connection to
the one-body quantum marginal problem, or $N$-representability problem in
quantum chemistry.
This fundamental problem about quantum correlations asks which
single-particle density matrices $\rho^{(1)}, \dots,
\rho^{(N)}$ can appear as the reduced density matrices of a
globally pure quantum state. %${\lvert\Psi\rangle}$.
Its solution is easily seen to depend only on the eigenvalues
$\vec \lambda^{(k)}$ of the densities and allows for an
elegant mathematical description: the set of possible
vectors $\vec\lambda = (\vec\lambda^{(1)}, \ldots, \vec\lambda^{(N)})$ forms a convex
polytope \cite{christandlmitchison06,klyachko04,daftuarhayden04} whose
defining inequalities can be computed algorithmically
\cite{klyachko04,daftuarhayden04}.
For fermions, the most famous such inequality is the Pauli principle \cite{colemanyukalov00,klyachko06}.

\begin{figure}
\includegraphics[width=\linewidth]{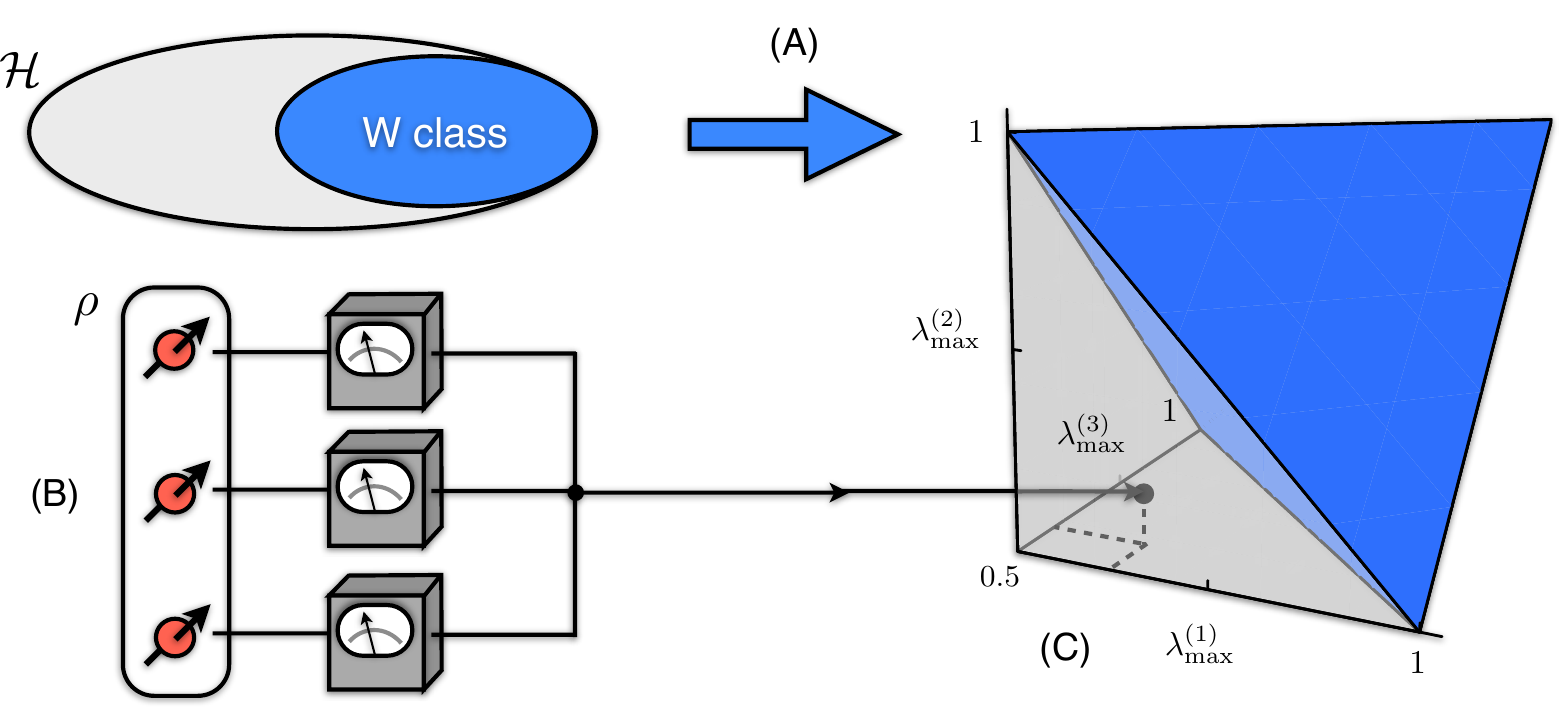}
\caption{{\it Entanglement polytopes as witnesses (illustrated for three qubits).}
(A) An entanglement polytope contains all possible local eigenvalues
of states in the entanglement class (the W class and its polytope
are shown in blue).
(B) For a sufficiently pure quantum state $\rho$, local tomography is performed to determine its local
eigenvalues.
(C) The indicated eigenvalues are not compatible with the W class,
hence $\rho$ must have GHZ-type entanglement.}
\end{figure}

Here we make the crucial observation that these local eigenvalues alone can already give considerable information about the entanglement of the global state, provided that it is pure. To make this precise, we consider the set of all eigenvalue vectors $\vec\lambda$ of the states in the closure of a given entanglement class. Surprisingly, this set also forms a convex polytope (i.e., it is the convex hull of finitely many such vectors), and we call it the entanglement polytope of the class.
Entanglement polytopes immediately lead to a local criterion for
witnessing global multi-particle entanglement: If the collection of
eigenvalues $\vec\lambda = (\vec\lambda^{(1)}, \ldots, \vec\lambda^{(N)})$ of the
single-particle reduced density matrices of a pure quantum state ${\lvert \psi \rangle}$
does not lie in an entanglement polytope $\Delta_{\mathcal C}$, then
the given state cannot belong to the corresponding entanglement class
$\mathcal C$ (Fig.~1). Mathematically,
\begin{equation}
  \label{criterion}
  \vec\lambda = (\vec\lambda^{(1)}, \ldots, \vec\lambda^{(N)}) \notin \Delta_\mathcal C
  \,\Longrightarrow\,
  {\lvert \psi \rangle} \notin \mathcal C.
\end{equation}
Phrased differently, the criterion allows us to witness the presence
of a highly entangled state by showing that its local eigenvalues are
incompatible with all less-entangled classes.
% This reasoning is easily seen to provide the strongest possible statements about the
% entanglement class of the state based on single-particle information alone and
% motivates the study of the polytopes.
%
% The inequalities describing an entanglement polytope can be seen as analogous to Bell inequalities. However, whereas the latter constrain the correlations that can be explained by joint measurements in local hidden-variable models, entanglement polytopes constrain the states attained simultaneously by the individual particles.
%
Strikingly, there are always only finitely many entanglement polytopes, and they naturally form a hierarchy: if a state in the class
$\mathcal C$ can be approximated arbitrarily well by states from $\mathcal D$ then
$\Delta_{\mathcal C} \subseteq \Delta_{\mathcal D}$.
This reflects geometrically the fact that states in the second class are more powerful for quantum information processing.

In order to compute $\Delta_{\mathcal C}$, and to see that it is
indeed a convex polytope, we use tools from algebraic geometry
and group representation theory, presented in detail in \cite{suppltext}.
We use the characterization of SLOCC operations as invertible local operators $A_1 \otimes \ldots \otimes A_N$ \cite{durvidalcirac00}, which act on the class $\mathcal{C}$, and therefore also on the set of polynomial functions on $\mathcal{C}$. The irreducible subspaces of
this action correspond to covariants, i.e.\ vector-valued
polynomial functions transforming in a well-defined way. By the
representation theory of Lie groups, each covariant is labeled by a
highest weight $\vec\mu = (\vec\mu^{(1)}, \ldots,
\vec\mu^{(N)})$, where the $\vec\mu^{(k)}$ are vectors of natural numbers,
whose entries are ordered decreasingly and sum to the degree $n$ of the polynomial.
Thus, any normalized highest weight $\vec\mu / n$ formally looks like an eigenvalue vector $\vec\lambda$.
This formal similarity corresponds to a factual correspondence:
%Building on work from algebraic geometry \cite{brion87}, we show that
$\Delta_{\mathcal{C}}$ is essentially given by those $\vec\mu / n$'s
whose associated covariants do not vanish on $\mathcal{C}$ \cite{brion87}.
The statement that $\Delta_{\mathcal C}$ is a convex polytope then
follows from the fact that the covariants form a finitely-generated
algebra \cite{brion87}.
We explain how to algorithmically compute a finite set of generators by
using computational invariant theory \cite{derksenkemper02},
motivated in part by \cite{dorankirwan07}.
The polytope can then be obtained as the
convex hull of the normalized highest weights $\vec\mu / n$
of those generators which are non-zero on the class $\mathcal C$. %any fixed quantum state ${\lvert \psi \rangle}$ in $\mathcal C$.

In the following we illustrate our method with a number of paradigmatic examples:
For qubit systems, each single-particle reduced density matrix $\rho^{(k)}$  has two eigenvalues, which are non-negative and sum to one; hence its spectrum is completely characterized by the maximal eigenvalue $\lambda^{(k)}_{\max}$, which can take values in the interval $[0.5,1]$.
In the case of three qubits, we may therefore regard the entanglement polytopes as subsets of three-dimensional space.
There are two full-dimensional polytopes \cite{han_compatible_2004}: one
for the W class (the upper pyramid in Fig.~1)
and the other for the GHZ class (the entire polytope, i.e., the
union of both pyramids). The tip of the upper pyramid
constitutes a polytope by itself, indicating a product state.
Three further one-dimensional polytopes are given by the edges
emanating from this vertex. They correspond to the three possibilities
of embedding a Bell state %and a single-qubit state
into three systems.
Thus, eigenvalues in the interior of the polytope are compatible only
with W and GHZ classes, i.e., genuine three-partite entanglement.
Likewise, if the
eigenvalues lie in the lower pyramid %, $\lambda^{(1)}_{\max} + \lambda^{(2)}_{\max} + \lambda^{(3)}_{\max} < 2$,
then by Eq.~\ref{criterion} the state cannot be contained in the closure of the $W$
class---we have witnessed GHZ-type entanglement.

\begin{figure}
\includegraphics[width=\linewidth]{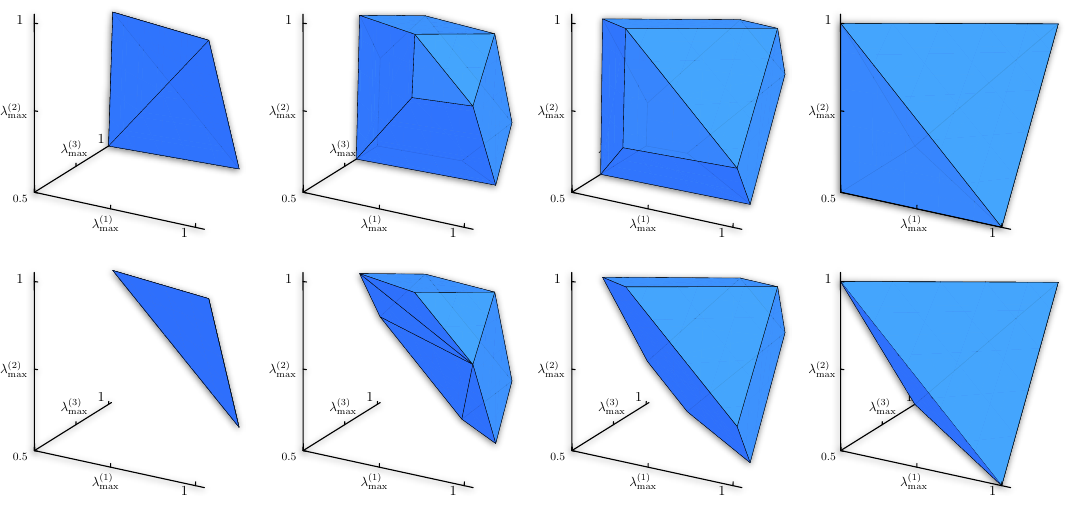}
\caption{{\it Cross-sections of two entanglement polytopes for four qubits.}
Each row shows cross-sections of a four-dim\-en\-sio\-nal entanglement polytope for four values of $\lambda^{(4)}_{\max}$: $0.5$, $0.66$, $0.83$ and $1$.
%Each row represents a single four-dimensional entanglement polytope.
%The eigenvalue $\lambda^{(4)}_{\max}$ is fixed in every column---to values
%$0.5$, $0.66$, $0.83$, and $1$, respectively.
The first row corresponds to the entanglement class $L_{a_4}$ for $a=0$ in the
characterization of \cite{verstraetedehaenedemooretal02} and
the second row to the four-qubit W class---one can clearly
identify the ``upper-pyramid
form'' explained in the text.
Several properties that had previously been computed algebraically \cite{verstraetedehaenedemooretal02}
can be read off directly: E.g.,
the final column corresponds to the situation when the fourth qubit
has been projected onto a pure state;
%we recover a state on the remaining three sites, which
as apparent from the polytopes,
the state of the remaining three sites is generically of GHZ type in the first row,
and of W type in the second row. See (26)
for an interactive visualization of all four-qubit entanglement polytopes.}
\end{figure}

In systems of 4 qubits, there exist 9 infinite families of entanglement
classes, each described by up to four complex parameters
\cite{verstraetedehaenedemooretal02} that are not directly
accessible; hence, a complete classification is too detailed
to be practical. In contrast, the polytope method strikes an attractive
balance between coarse-graining and preserving structure
(Fig.~2):
Up to permutations, there are 12 entanglement polytopes, 7 of which
are full-dimensional and correspond to distinct types of genuine
four-partite entanglement. One example is the 4-qubit W class: in complete analogy to the previous case,
its polytope is an ``upper pyramid'' of eigenvalues that fulfill
$
  \lambda^{(1)}_{\max} + \lambda^{(2)}_{\max} +
  \lambda^{(3)}_{\max} + \lambda^{(4)}_{\max} \geq 3
$.

%\emph{$N$ Qubits}.---%
Quan\-tum states which are genuinely multipartite entangled are of
particular interest \cite{guehne_multipartite_2005}. These are the states
which do not factorize with respect to any partition of the qubits into two sets.
We show for arbitrary qubit systems that the entanglement polytopes of the biseparable states (i.e., the states that do factorize) do not account for all possible eigenvalues. Therefore, the presence of genuine multipartite entanglement in a pure quantum state can be certified by checking that the local eigenvalues do not lie in any biseparable polytope (Fig.~3).

\begin{figure}
\includegraphics[width=0.5\linewidth]{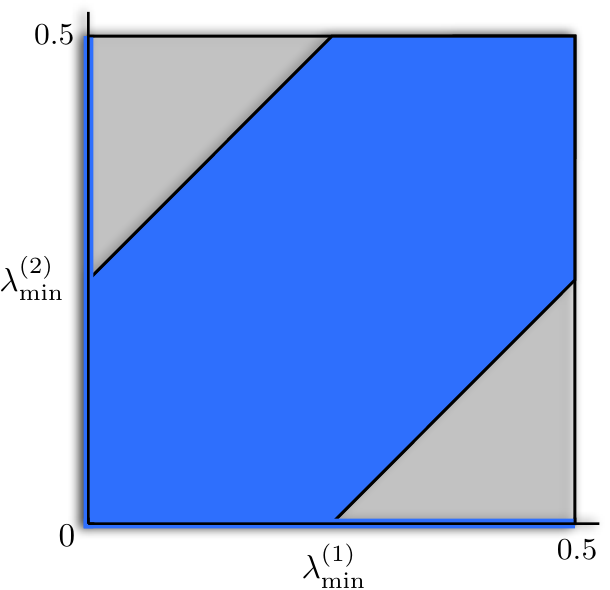}
\caption{{\it Witnessing genuine multipartite entanglement.}
The figure displays the two-dimensional cross-section through the six-qubit eigenvalue polytope where we fix $\lambda^{(3)}_{\min} = \ldots = \lambda^{(6)}_{\min} = 0.125$.
If the local eigenvalues do not belong to any biseparable entanglement polytope (blue region) then any corresponding pure state necessarily contains genuine six-qubit entanglement.}
\end{figure}

%\emph{Bosons and Fermions}.---%
Entanglement polytopes can also be constructed for quantum systems composed of bosons or fermions. % fixed particle number?
Since the individual particles are indistinguishable, the reduced density matrices $\rho^{(k)}$ and their eigenvalues $\vec\lambda^{(k)}$ all coincide.
In the case of qubits, the entanglement polytopes are therefore intervals describing the possible values of the maximal local eigenvalue for states in the class.
For bosons, the right endpoint of the interval is always equal to one, corresponding to a coherent state, while the left endpoint $\gamma_{\mathcal C}$ corresponds to the most entangled states in the class.
These are the symmetric Dicke states with mean spin $m \geq 0$ per particle, for which $\gamma_C = 0.5 + m$, and they include the GHZ state for $\gamma_C = 0.5$.
In \cite{suppltext} we also describe the entanglement polytopes for the Borland--Dennis system, which is composed of three fermions with local rank six \cite{borlanddennis72}.

%Entanglement polytopes can also be used to
We can also obtain quantitative information about the multi-particle entanglement of a quantum state, e.g., by witnessing genuine $k$-partite entanglement \cite{guehne_multipartite_2005} using a generalization of the argument sketched above.
Furthermore, we may consider the linear entropy of entanglement
$E(\psi) = 1 - \frac{1}{N}\sum_{j=1}^N \mathrm{tr}(\rho^{(j)})^2$
used, e.g., in metrology
\cite{furuyanemespellegrino98}. %,barnumknillortizetal04}. %,boixomonras08}.
Entanglement polytopes allow us to bound the maximal linear entropy of entanglement
distillable % from the global state
by SLOCC operations. %, even though only an exponentially small fraction of the state is known.
Since $E(\psi)$ corresponds to the Euclidean length of the vector of local eigenvalues, % $\vec\lambda(\psi)$:
shorter vectors imply more entanglement.  In particular, quantum
states of maximal entropy of entanglement in a class $\mathcal C$ map
to the point of minimal distance to the origin in the entanglement
polytope $\Delta_{\mathcal C}$. Therefore, if the local eigenvalues of
a given state lie only in polytopes with small distance to the origin,
a high amount of entanglement can be distilled.
In \cite{suppltext} we turn this observation into a quantitative statement
and describe a distillation algorithm.

%For pure quantum states, the described methods are efficient: estimating the eigenvalue vector requiring only a linear number of single-particle measurements. Moreover, they are robust with respect to measurement errors, since the eigenvalues change continuously with the global state.

%
% All methods introduced in this paper are resilient against
% not-too-large levels of noise and therefore applicable to quantum
% states $\rho$ prepared in the laboratory.
% Intuitively, the effect of
% noise is to ``wash-out'' local spectra.
%

Quantum states prepared in the laboratory are always subject to noise and hence never perfectly pure.
%Intuitively, the effect of noise is to perturb the local eigenvalues (i.e., it effectively enlarges the error bars).
Our method for witnessing entanglement using Eq.~\ref{criterion} can be adapted to this situation as long as the noise is not too large.
To make this statement precise, we assume that a lower bound $1-\varepsilon$ on the purity $\mathrm{tr}\,\rho^2$ of a quantum state $\rho$ is available. This implies that $\rho$ has fidelity $\langle \psi | \rho | \psi \rangle \geq 1-\varepsilon$ with a pure state $\lvert\psi\rangle$ whose local eigenvalues deviate from the measured ones by no more than a small amount $\delta(\varepsilon)$.
In the case of $N$ qubits one has that $\delta(\varepsilon) \approx N \varepsilon / 2$ for small impurities.
Therefore, as long as the distance of the measured eigenvalues $\vec\lambda$ to the entanglement polytope $\Delta_{\mathcal C}$ is at least $\delta(\varepsilon)$, the experimentally prepared state $\rho$ has high fidelity with a pure state that is more entangled than the class $\mathcal C$.
These ideas can be further extended to show that $\rho$ itself cannot be written as a convex combination of quantum states in a given entanglement class, as we describe in \cite{suppltext}.
%
% Whereas the local eigenvalues $\vec\lambda$ can be measured by single-particle tomography,
Unlike the local eigenvalues, the purity $\mathrm{tr}\,\rho^2$ cannot be determined by single-particle tomography alone.
However, it is in general not necessary to perform full tomography of the global state in order to estimate the purity \cite{buhrmanclevewatrousetal01}.
Whereas it is an experimental challenge to achieve the levels of purity necessary for the application of our method, we believe that they are in the reach of current technology \cite{rauschenbeutelnoguesosnaghietal00,pandaniellgasparonietal01,monzschindlerbarreiroetal11}.

\nocite{polytopes12}

\bibliography{slocc}

\bigskip\noindent{\bf Acknowledgements:} % begin acknowledgements
We thank A.~Botero, F.~Brand\~ao,
P.~B\"urgisser, C.~Ikenmeyer, F.~Kirwan, F.~Mintert, G.~Mitchison, A.~Osterloh, and P.~Vrana for
valuable discussions.
This work is supported by
the German Science Foundation (grant CH 843/2--1),
the Swiss National Science Foundation (grants PP00P2\_128455, 20CH21\_138799 (CHIST-ERA project CQC) and 200021\_138071),
the Swiss National Center of Competence in Research `QSIT',
and the Excellence Initiative of the German Federal and State Governments (grant ZUK 43).
After completion of this work, we have learned about independent related work by Sawicki, Oszmaniec and Ku\'{s}.
% end acknowledgements

\bigskip

\noindent {\bf Supplementary Materials:}\\
\href{http://www.sciencemag.org/content/340/6137/1205/suppl/DC1}{www.sciencemag.org/content/340/6137/1205/suppl/DC1}\\
Supplementary Text\\
Figs.~S1--S4\\
Tables~S1--S3\\
References (27--83)

\cleardoublepage
\onecolumngrid
\appendix*
\setcounter{page}{1}
\setcounter{equation}{0}
\setcounter{figure}{0}
\renewcommand{\thetable}{S\arabic{table}}
\renewcommand{\thefigure}{S\arabic{figure}}
\renewcommand{\theequation}{S\arabic{equation}}

\section*{Supplementary Text}

\subsection{Introduction}

In this supplement, we will give rigorous proofs of the main properties of entanglement polytopes (\autoref{entanglement polytopes}) and describe an algorithmic method for their computation, which we illustrate with several worked examples (\autoref{examples}).
Moreover, we elaborate on the properties of the linear entropy of entanglement and derive the gradient flow procedure for entanglement distillation that has been sketched in the main body of the article (\autoref{entropy of entanglement}). Our main technical tools are Brion's invariant-theoretic description of moment polytopes \cite{brion87} seen through the lens of non-reductive group actions \cite{dorankirwan07}, and Kirwan's analysis of the equivariant Morse gradient flow for the norm square of a moment map \cite{kirwan84}.

For clarity of exposition, we will describe our results for a system composed of $N$ distinguishable subsystems with $d_k$ degrees of freedom, respectively, and Hilbert space
$$\mathcal H = \CC^{d_1} \otimes \ldots \otimes \CC^{d_N}.$$
In the context of multi-particle entanglement, we will think of each of the $N$ subsystems as corresponding to an individual particle. However, the subsystems can be of more general nature and, e.g., describe different degrees of freedom such as position and spin.
All results can be adapted to systems composed of bosons or fermions by replacing $\mathcal H$ with the (anti)symmetric subspace (\autoref{examples}).
We shall denote by
$$\PP(\mathcal H) = \{ \rho = \ketbra \psi \psi : \braket{\psi | \psi} = 1 \}$$
the projective space of \emph{pure states}; the expectation value of an observable $O$ in a pure state $\rho$ is given by $\tr(\rho \, O) = \braket{\psi | O | \psi}$.
The effective state of the $k$-th particle is described by the \emph{(one-body) reduced density matrix} $\rho^{(k)}$, which by definition reproduces the expectation values of all one-body observables $O^{(k)}$,
\begin{equation}
  \label{RDM}
  \tr( \rho^{(k)} O^{(k)})
  = \tr(\rho \, ( \id^{\otimes (k-1)} \otimes O^{(k)} \otimes \id^{\otimes(N-k)} ) )
  = \braket{\psi | \id^{\otimes (k-1)} \otimes O^{(k)} \otimes \id^{\otimes(N-k)} | \psi}
\end{equation}
The one-body reduced density matrices $\rho^{(k)}$ are in general \emph{mixed states}, i.e., convex mixtures of pure states. Formally, they are positive semidefinite Hermitian operators of trace one. By the spectral theorem, each $\rho^{(k)}$ can be diagonalized, and the \emph{local eigenvalues} $\vec\lambda^{(k)} = \vec\lambda(\rho^{(k)})$ are vectors of non-negative numbers summing to one (which we assume by convention to be weakly decreasing).

\subsection{Multi-Particle Entanglement and Stochastic Local Operations and Classical Communication (SLOCC)}
\label{slocc}

A pure quantum state $\rho = \ketbra \psi \psi$ of a multi-particle system is called \emph{entangled} if it cannot be written as a tensor product \cite{nielsenchuang04}
\begin{equation*}
  \ket \psi \neq \ket{\psi^{(1)}} \otimes \ldots \otimes \ket{\psi^{(N)}}.
\end{equation*}
It is easy to see that $\rho$ is unentangled, or \emph{separable}, if and only if all its one-body reduced density matrices $\rho^{(k)}$ are pure states, that is, if and only if $\vec\lambda^{(k)} = (1,0,\ldots,0)$ for $k=1, \ldots, N$.
%In the latter case, $\rho^{(k)} = \ketbra {\psi^{(k)}} {\psi^{(k)}}$.

In order to classify the entanglement present in a given multipartite quantum state $\rho$, it is useful to compare its capability for quantum information processing tasks with that of other quantum states. Specifically, we shall think of $\rho$ to be at least as entangled as any other quantum state $\rho'$ that can be produced from a single copy of $\rho$ by performing a sequence of \emph{stochastic local operations} (i.e., local operations which succeed with some positive probability, e.g., by post-selecting on a certain measurement outcome) \emph{and classical communication (SLOCC)}
%, since clearly $\rho$ can be used as a replacement for $\rho'$ in any quantum information processing protocol
\cite{bennettpopescurohrlichetal00,durvidalcirac00}. If, conversely, $\rho$ can also be produced from $\rho'$ then we can think of the two states as possessing the same kind of multi-particle entanglement. In this way the set of quantum states is partitioned into equivalence classes.
For pure states $\rho = \ketbra \psi \psi$ and $\rho' = \ketbra {\psi'} {\psi'}$, it has been shown that they are equivalent under SLOCC if and only if there exist invertible operators $A^{(k)}$ such that
\begin{equation*}
  \ket{\psi'} = ( A^{(1)} \otimes \ldots \otimes A^{(N)} ) \ket\psi.
\end{equation*}
Indeed, the operators $A^{(k)}$ can be defined by following a successful branch of a conversion protocol. Conversely, given operators $A^{(k)}$ it suffices to perform successive local POVM measurements with Kraus operators
\begin{equation}
  \label{SL POVM}
  S^{(k)} = \sqrt{p^{(k)}} A^{(k)}, ~ F^{(k)} = \sqrt{\id - p^{(k)} (A^{(k)})^\dagger A^{(k)}},
\end{equation}
where $p^{(k)}$ is a suitable normalization constant. If all measurements succeed (outcomes $S^{(1)}$, \ldots, $S^{(k)}$) then $\rho'$ has been successfully distilled from $\rho$ \cite{durvidalcirac00}.
Mathematically, it is convenient to define an action of the Lie group $G = \SL(d_1) \times \ldots \times \SL(d_N)$ of invertible local operators of determinant one on the projective space $\PP(\mathcal H)$ by
\begin{equation}
  \label{SL action}
  g \cdot \ketbra \psi \psi := \frac { ( g^{(1)} \otimes \ldots \otimes g^{(N)} ) \ketbra \psi \psi ( g^{(1)} \otimes \ldots \otimes g^{(N)} )^\dagger } {\norm{ ( g^{(1)} \otimes \ldots \otimes g^{(N)} ) \ket \psi}^2},
\end{equation}
where $g = (g^{(1)}, \ldots, g^{(N)}) \in G$. Then the result of \cite{durvidalcirac00} states that two pure states $\rho$ and $\rho'$ are equivalent under SLOCC if and only if $G \cdot \rho = G \cdot \rho'$. In other words, the \emph{SLOCC entanglement class} containing $\rho$ is just the orbit $\mathcal C = G \cdot \rho$. Clearly, the unentangled states form a single SLOCC class.

The closure $\overline{\mathcal C} = \overline{G \cdot \rho}$ of an entanglement class contains in addition those quantum states which can be arbitrarily well approximated by a state in the class.
%Equivalently, these are the quantum states that can be produced from $\rho$ by (not necessarily invertible) SLOCC operations up to arbitrarily small noise.
%For typical quantum information processing protocols, which are continuous and therefore robust with respect to slight noise, these states can therefore be substituted by $\rho$ (but not vice versa).
In this way, the closure of an entanglement class can be given a similar operational interpretation as the class itself. While the entanglement classes partition the set of multi-particle quantum states, their closures naturally form a hierarchy. This is because they are stable under SLOCC operations; indeed, it is immediate that $\rho' \in \overline{G \cdot \rho}$ implies $\overline{G \cdot \rho'} \subseteq \overline{G \cdot \rho}$.
In particular, every entanglement class contains in its closure the class of unentangled states.

In summary, stochastic local operations and classical communication provide a systematic way of studying multi-particle entanglement. However, it is immediate from the fact that the dimension of $G$ grows only linearly with $N$ that there is generically an infinite number of distinct SLOCC entanglement classes labeled by an exponential number of continuous parameters (cf.~\autoref{examples}). It is therefore necessary to coarsen the classification in a systematic way in order to arrive at a tractable way of witnessing multi-particle entanglement.

We note that an extraordinary amount of research has been devoted to
the task of classifying SLOCC entanglement and identifying it
experimentally. The field is far too large to allow for an exhaustive
bibliography. For reviews on the general theory, see
\cite{horodecki09,eisert_multiparticle_2008}; a review focussing on
detection is \cite{guhne_entanglement_2009}. Methods from algebraic geometry and classical
invariant theory have long been used to analyze entanglement classes,
see, e.g., \cite{
verstraetedehaenedemooretal02,
klyachko_coherent_2002,
briandluquethibon03,
verstraetedehaenedemoor03,
miyake03,
leifer_measuring_2004,
han_compatible_2004,
osterloh_constructing_2005,
osterlohsiewert06,
klyachko07,
dokovicosterloh09,
bastinkrinsmathonetetal09,
osterloh10,
gour_necessary_2011,
viehmanneltschkasiewert11,
eltschka_multipartite_2012}
and references therein.

\subsection{Entanglement Polytopes}
\label{entanglement polytopes}

\subsubsection{Definition}

The \emph{entanglement polytope} of an entanglement class $\mathcal C$ is by definition
\begin{equation*}
  \Delta_{\mathcal C} =
  \left\{
    (\vec\lambda^{(1)}, \ldots, \vec\lambda^{(N)})
    \,:\,
    \vec\lambda^{(k)} = \vec\lambda(\rho^{(k)}),\,
    \rho \in \overline{\mathcal C}
  \right\},
\end{equation*}
the set of tuples of local eigenvalues of the quantum states in the closure of the entanglement class.
We will show below that this set is in fact a convex polytope.
The set of entanglement polytopes forms a hierarchy which coarsens the hierarchy of the closures of entanglement classes described above. Namely,
\begin{equation*}
  \overline{\mathcal C} \subseteq \overline{\mathcal D}
  \,\Longrightarrow\,
  \Delta_{\mathcal C} \subseteq \Delta_{\mathcal D}.
\end{equation*}
In particular, if a pure quantum state $\rho = \ketbra \psi \psi$ is contained in an entanglement class $\mathcal C$ then its collection of local eigenvalues is contained in the corresponding entanglement polytope $\Delta_{\mathcal C}$. That is,
\begin{equation*}
  \rho \in \mathcal C
  \,\Longrightarrow\,
  (\vec\lambda^{(1)}, \ldots, \vec\lambda^{(N)}) \in \Delta_{\mathcal C}.
\end{equation*}
Equivalently, if the collection of local eigenvalues is not contained in the entanglement polytope then the quantum state is necessarily in a different entanglement class:
\begin{equation}
  \label{appendix criterion}
  (\vec\lambda^{(1)}, \ldots, \vec\lambda^{(N)}) \notin \Delta_{\mathcal C}
  \,\Longrightarrow\,
  \rho \notin \mathcal C
\end{equation}
This establishes our main criterion for witnessing multi-particle entanglement.

\subsubsection{Invariant-Theoretic Description}
\label{covariants}

In order to analyze the properties of entanglement polytopes, it is useful to introduce the map
\begin{equation*}
  \Phi \colon
  %\PP(\mathcal H) \rightarrow \CC^{d_1 \times d_1} \times \ldots \times \CC^{d_N \times d_N},~
  \rho = \ketbra \psi \psi \mapsto (\rho^{(1)}, \ldots, \rho^{(N)})
\end{equation*}
which assigns to a quantum state the collection of its one-body reduced density matrices.
Given a product of local unitaries $U = U^{(1)} \otimes \ldots \otimes U^{(N)}$, it follows from \eqref{RDM} that
\begin{equation}
  \label{RDM equivariance}
  \Phi( U \rho U^\dagger) =
  ( U^{(1)} \rho^{(1)} (U^{(1)})^\dagger, \ldots, U^{(N)} \rho^{(N)} (U^{(N)})^\dagger ).
\end{equation}
We can therefore jointly diagonalize the reduced density matrices by applying suitable local unitaries. On the other hand, we observe that the action $\rho \mapsto U \rho U^\dagger$ is simply the restriction of \eqref{SL action} to the subgroup of local unitaries $K = \SU(d_1) \times \ldots \times \SU(d_N)$ of the group $G = \SL(d_1) \times \ldots \times \SL(d_N)$ introduced in \autoref{slocc} (the denominator is equal to unity since any unitary is norm-preserving), so that the entanglement class is left unchanged.
We conclude that for every quantum state $\rho \in \mathcal C$ there exists a quantum state $\rho'$ in the same entanglement class whose reduced density matrices are diagonal in the computational basis. As we may identify diagonal density matrices with their collection of eigenvalues (for definiteness, we shall require the diagonal entries to be arranged in weakly decreasing order), each entanglement polytope can be written as an intersection
\begin{equation}
  \label{moment polytope}
  \Delta_{\mathcal C} = \Phi(\overline{\mathcal C}) \cap D_\downarrow,
\end{equation}
where we denote by $D_\downarrow$ the set of tuples of diagonal density matrices with weakly decreasing entries.

In mathematical terms, the map $\Phi$ can be identified with the canonical \emph{moment map} for the action of $K$ on the projective space $\PP(\mathcal H)$ \cite{guilleminsternberg84,kirwan84}. Indeed, by \eqref{RDM equivariance} it is equivariant with respect to the action of $K$ on tuples of reduced density matrices, which can be identified with the coadjoint action of $K$ on the dual of its Lie algebra, and its components are Hamiltonian functions with respect to the Fubini--Study symplectic form $\omega_{\text{FS}}$,
\begin{equation}
  \label{moment map equation}
  d\braket{\Phi, X}(\widehat{i Y}_{\rho})
  = i \tr([Y,\rho] X)
  = i \tr(\rho [X,Y])
  = \omega_{\text{FS}}(\widehat{i X}_{\rho}, \widehat{i Y}_{\rho}),
  %= \tr(\widehat{X}_{\rho} \, \widehat{i Y}_{\rho}),
\end{equation}
where $X$ and $Y$ are Hermitian operators on $\mathcal H$ and where $\widehat Z_\rho = \restrict{\frac d {dt}}{t=0} e^{tZ} \cdot \rho = \restrict{\frac d {dt}}{t=0} e^{tZ} \ketbra \psi \psi e^{tZ^\dagger} /\, \norm{e^{tZ} \ket \psi}_2^2$ is the tangent vector generated by the infinitesimal action of an arbitrary operator $Z$; in particular, $\widehat{i X}_\rho = i[X,\rho]$ and $\widehat{X}_\rho = \{X,\rho\} - 2 \tr(\rho X) \rho$. Formally, the above equation should be restricted to local observables $X$ of trace zero, since only these correspond to elements in the Lie algebra of the group $K$.

It follows that \eqref{moment polytope} is just the definition of the \emph{moment polytope}, or \emph{Kirwan polytope}, for the subvariety $\overline{\mathcal C} \subseteq \PP(\mathcal H)$, with $D_\downarrow$ corresponding to a choice of positive Weyl chamber.
It is a celebrated result in mathematics that the moment polytope of a compact symplectic manifold is a convex polytope \cite{kirwan84b}.
For the singular case of $G$-stable irreducible projective subvarieties of projective spaces, which is a category of spaces that includes the closures of any entanglement class, this has been established by Brion \cite{brion87}.

Brion's proof of convexity relies on an alternative, invariant-theoretic description of the moment polytopes. In order to apply his results to the case of entanglement polytopes, we need some notions from representation theory and invariant theory. Recall that the finite-dimensional polynomial irreducible representations of the group $\GL(d)$ of invertible $d \times d$-matrices are labeled by integer vectors $\vec\mu \in \ZZ^d_{\geq 0}$ with weakly decreasing entries---their highest weight, or Young diagram \cite{fulton97}. The irreducible representations of $\SL(d)$ can be obtained by restriction, and they will be denoted by $V^d_{\vec\mu}$.
%It can be shown that $V^d_{\vec\mu}$ is already determined by the differences $\mu_{j+1}-\mu_j$.
Finally, a \emph{covariant} of degree $n$ and weights $\vec\mu^{(1)}, \ldots, \vec\mu^{(N)}$ is by definition a $G$-equivariant map
\begin{equation*}
  \mathcal H = \CC^{d_1} \otimes \ldots \otimes \CC^{d_N} \longrightarrow V^{d_1}_{\vec\mu^{(1)}} \otimes \ldots \otimes V^{d_N}_{\vec\mu^{(N)}}
\end{equation*}
whose components are homogeneous polynomials of degree $n$. Note that the right-hand side is a $G$-representation in the obvious way. %; moreover, any irreducible representation of $G$ arises as such a tensor product.
We shall require each $\vec\mu^{(k)}$ to satisfy $\sum_j \mu^{(k)}_j = n$; this can always be satisfied and does not impose an additional restriction. Moreover, we observe that even though covariants are defined on the Hilbert space $\mathcal H$, their non-vanishing can be well-defined on subsets of the projective space $\PP(\mathcal H)$ by taking any representative vector, since they are homogeneous functions; we will do so without further mention.
We can now specialize the main result of \cite{brion87} to our setting:

\begin{thm}
  \label{main thm}
  The entanglement polytope of an entanglement class $\mathcal C = G \cdot \rho$ is equal to
  \begin{align*}
    \Delta_{\mathcal C}
    % = &\overline{\left\{
    %   \frac 1 n \left( \vec\mu^{(1)}, \ldots, \vec\mu^{(N)} \right) :
    %   \text{$\exists$ a covariant of degree $n$ and weights $\vec\mu^{(k)}$ which does not vanish completely on $\overline{\mathcal C}$}
    % \right\}} \\
    = &\overline{\left\{
      \frac 1 n \left( \vec\mu^{(1)}, \ldots, \vec\mu^{(N)} \right) :
      \text{$\exists$ a covariant of degree $n$ and weights $\vec\mu^{(k)}$ which is non-zero at $\rho$}
    \right\}}
  \end{align*}
\end{thm}
\begin{proof}
  By the above discussion and \eqref{moment polytope}, we have seen that $\Delta_{\mathcal C}$ is equal to the moment polytope of the $G$-orbit closure $\overline{\mathcal C}$. In \cite{brion87} it is shown that this moment polytope can be described as the closure of the set of points $\left( \vec\mu^{(1)}, \ldots, \vec\mu^{(N)} \right)/n$ for which there exists a non-zero $G$-equivariant morphism
  \begin{equation*}
    \varphi \colon \widetilde{\mathcal C} \rightarrow V^{d_1}_{\vec\mu^{(1)}} \otimes \ldots \otimes V^{d_N}_{\vec\mu^{(N)}}
  \end{equation*}
  which is homogeneous of degree $n$ (i.e., whose components are homogeneous polynomial functions of degree $n$). Here, we denote by $\widetilde{\mathcal C} = \{ \ket \psi : \ketbra \psi \psi \; / \braket{\psi | \psi} \in \overline{\mathcal C} \}$ the affine cone over $\overline{\mathcal C}$.
  Evidently, every covariant which is non-zero at $\rho$ restricts to a map on $\widetilde{\mathcal C}$ which does not vanish completely; this shows one inclusion.
  For the converse, we first observe that every non-zero map $\varphi$ as above is automatically non-zero at $\rho$, for it is $G$-equivariant and the affine cone over $\mathcal C = G \cdot \rho$ is an open subset of $\widetilde{\mathcal C}$. It remains to show that it can be extended to a $G$-equivariant map on all of $\mathcal H$. This can be seen as follows by using some basic algebraic geometry and representation theory \cite{hartshorne77,knapp02}:

  Since the map $\varphi$ is $G$-equivariant and its range an irreducible representation, it is completely equivalent to consider instead of $\varphi$ its component $\phi = \braket{\varphi, v}$ evaluated at a highest weight vector $v$ with respect to some maximal unipotent subgroup $U \subset G$. The function $\phi$ is a $U$-invariant homogeneous polynomial function of degree $n$, i.e., a $U$-invariant element of the degree-$n$ piece of the homogeneous coordinate ring
  \begin{equation*}
    \CC[\widetilde{\mathcal C}]_n = \CC[\mathcal H]_n / I_n,
  \end{equation*}
  where $\CC[\mathcal H]$ denotes the algebra of polynomial functions on $\mathcal H$,
  $\CC[\mathcal H]_n$ the subspace of homogeneous functions of degree $n$, and
  $I_n = \{ p \in \CC[\mathcal H]_n : p(\widetilde{\mathcal C}) = 0 \}$ the corresponding part of the vanishing ideal of $\widetilde{\mathcal C}$. By averaging over the maximal compact subgroup $K \subseteq G$, we can find an orthogonal complement $I_n^\perp$ to $I_n$ that is $G$-invariant (this is oftentimes called Weyl's unitary trick). Therefore, the natural map
  \begin{equation*}
    I_n^\perp \subseteq \CC[\mathcal H]_n \rightarrow \CC[\widetilde{\mathcal C}]_n
  \end{equation*}
  is a $G$-isomorphism, and it can be used to find a $U$-invariant extension of $\phi$ to all of $\mathcal H$. By reversing the above procedure, we obtain a corresponding covariant which $G$-equivariantly extends $\varphi$ to all of $\mathcal H$.
\end{proof}

\subsubsection{Properties and Computation}
\label{properties and computation}

The set of covariants of fixed degree and weight form a vector space. Moreover, given any two covariants $P$, $Q$ of degree $n$,$m$ and weights $\vec\mu^{(k)}$, $\vec\nu^{(k)}$, we can form their product by defining
\begin{equation}
  \label{multiplication}
  (P \circ Q)(\ket \psi) = (\pi^{(1)} \otimes \ldots \otimes \pi^{(N)}) \left( P(\ket \psi) \otimes Q(\ket \psi) \right).
\end{equation}
where $\pi^{(k)} \colon V^{d_k}_{\vec\mu^{(k)}} \otimes V^{d_k}_{\vec\nu^{(k)}} \rightarrow V^{d_k}_{\vec\mu^{(k)}+\vec\nu^{(k)}}$ denotes the projection onto the unique irreducible representation in the tensor product of two irreducible representations whose highest weight is the sum of the highest weights of the factors. The function $P \circ Q$ is a covariant of degree $n+m$ and weights $\vec\mu^{(k)} + \vec\nu^{(k)}$. Note that the point corresponding to $P \circ Q$,
\begin{equation}
  \label{convex}
  \frac {1} {n+m} \left( \vec\mu^{(1)} + \vec\nu^{(1)}, \ldots, \vec\mu^{(N)} + \vec\nu^{(N)} \right)
  = \frac {n} {n+m} \left( \frac 1 n \left( \vec\mu^{(1)}, \ldots, \vec\mu^{(N)} \right) \right)
  + \frac {m} {n+m} \left( \frac 1 m \left( \vec\nu^{(1)}, \ldots, \vec\nu^{(N)} \right) \right),
\end{equation}
is a convex combination of the points corresponding to the covariants $P$ and $Q$. It follows that the entanglement polytopes are convex (replace $P$ and $Q$ by their respective powers).
% eigentlich sollte hier noch stehn dass man so alle punkte auf dem ray kriegt...
If we allow for formal linear combinations of covariants of different degree or weight, we thus obtain a graded algebra, and it is well-known that this algebra is \emph{finitely generated}, i.e., there exist finitely many covariants $P_1, \ldots, P_m$ such that any other covariant can be expressed as a linear combination of monomials in the $P_j$.
In the literature, this is typically phrased as the well-known statement that the algebra $\CC[\mathcal H]^U$ of $U$-invariant polynomials is finitely generated \cite{kraft85,grosshans97}; \eqref{multiplication} then corresponds to the ordinary multiplication of functions.
As noted in \cite{brion87}, finite generation implies that moment polytopes are compact convex polytopes. Using \autoref{main thm}, we can slightly sharpen this result for the case of entanglement polytopes:

\begin{thm}
  \label{main cor}
  The entanglement polytope of an entanglement class $\mathcal C = G \cdot \rho$ is given by the convex hull
  \begin{align*}
    \Delta_{\mathcal C}
    % = &\overline{\left\{
    %   \frac 1 n \left( \vec\mu^{(1)}, \ldots, \vec\mu^{(N)} \right) :
    %   \text{$\exists$ a covariant of degree $n$ and weights $\vec\mu^{(k)}$ which does not vanish completely on $\overline{\mathcal C}$}
    % \right\}} \\
    = &\conv \left\{
      \frac 1 {n_j} \left( \vec\mu^{(1)}_j, \ldots, \vec\mu^{(N)}_j \right) :
      P_j(\rho) \neq 0
    \right\},
  \end{align*}
  where we denote by $n_j$ the degree and by $\vec\mu^{(k)}_j$ the weights of a generator $P_j$ ($j=1,\ldots,m$).
  In particular, $\Delta_{\mathcal C}$ is a compact convex polytope.
\end{thm}
\begin{proof}
  Consider a monomial $P_{j_1} \circ \cdots \circ P_{j_l}$ which does not vanish on $\rho$.
  Its degree is $n = \sum_i n_{j_i}$ and its weights are given by $\vec\mu^{(k)} = \sum_i \vec\mu^{(k)}_{j_i}$.
  Since necessarily $P_{j_i}(\rho) \neq 0$ for all $i=1,\ldots,l$, the corresponding point
  \begin{equation*}
    \frac 1 n \left( \vec\mu^{(1)}_j, \ldots, \vec\mu^{(N)}_j \right)
    = \sum_i \frac {n_{j_i}} n \left( \frac 1 {n_{j_i}} \left( \vec\mu^{(1)}_{j_i}, \ldots, \vec\mu^{(N)}_{j_i} \right) \right)
  \end{equation*}
  is contained in the convex hull displayed in the statement of the corollary.

  If $P$ is an arbitrary covariant of degree $n$ and weights $\vec\mu^{(k)}$ which does not vanish on $\rho$, we can write it as a linear combination of monomials. If the linear combination is chosen minimally then every monomial has the same degree and weights as $P$. Since at least one of the monomials must not vanish on $\rho$, the claim follows from what we have proved above.
\end{proof}

By virtue of \autoref{main cor}, the \emph{computation} of entanglement polytopes is a finite problem that can be completely algorithmized. Indeed, by using the relation \cite[\S{}4.2]{dolgachev03} (cf.~\cite{dorankirwan07})
\begin{equation*}
  \CC[\mathcal H]^U \cong
  \left( \CC[\mathcal H] \otimes \CC[\overline{G/U}] \right)^G \cong
  \CC[\mathcal H \times \overline{G/U}]^G
\end{equation*}
the problem of computing a set of generating covariants $P_1, \ldots, P_m$ is transformed into a problem of computing invariants for a complex reductive group, for which there exist algorithms in computational invariant theory \cite{derksenkemper02}. %, as implemented e.g.~in the Magma computer algebra system.
Once a set of generators has been found, \autoref{main cor} can be used to compute the entanglement polytope both for specific states as well as for families of states. We use this method to compute all examples in \autoref{examples}.

Finite generation also implies other desirable properties.
It is clear that there are only \emph{finitely many} entanglement polytopes, since by \autoref{main cor} any entanglement polytopes is the convex hull of some subset of the finite set of points
\begin{equation*}
  \mathcal P = \{ ( \vec\mu^{(1)}_j, \ldots, \vec\mu^{(N)}_j ) / n_j : j=1,\ldots,m\}
\end{equation*}
Moreover, as similarly observed in \cite{brion87}, the set of quantum states for which all generators are non-zero,
\begin{equation*}
  \bigcap_{j=1}^m \{ \rho \in \PP(\mathcal H) : P_j(\rho) \neq 0 \},
\end{equation*}
is a finite intersection of Zariski-open sets, hence itself Zariski-open. It follows that the entanglement polytope of a generic quantum state is maximal, i.e., equal to the convex hull of the finite set $\mathcal P$.

Let us briefly digress to discuss this generic entanglement polytope, which we shall denote by $\Delta$. Clearly,
\begin{align*}
  \Delta &= \conv \left\{
    \frac 1 {n_j} \left( \vec\mu^{(1)}_j, \ldots, \vec\mu^{(N)}_j \right)
  \right\} \\
  &= \overline{\left\{
    \frac 1 n \left( \vec\mu^{(1)}, \ldots, \vec\mu^{(N)} \right) :
    \text{$\exists$ a covariant of degree $n$ and weights $\vec\mu^{(k)}$}
  \right\}} \\
  &= \Phi(\PP(\mathcal H)) \cap D_\downarrow \\
  &= \left\{
    (\vec\lambda^{(1)}, \ldots, \vec\lambda^{(N)})
    \,:\,
    \vec\lambda^{(k)} = \vec\lambda(\rho^{(k)}),\,
    \rho \in \PP(\mathcal H)
  \right\},
\end{align*}
the set of possible local eigenvalues of an arbitrary pure quantum state (not restricted to any particular entanglement class).
The problem of computing this polytope is known as the \emph{one-body
quantum marginal problem} in quantum information theory and as the
one-body $N$-representability problem in quantum chemistry
\cite{colemanyukalov00}. Its convexity has been noted in
\cite{christandlmitchison06,daftuarhayden04,klyachko04}, and it has
been solved by combining the invariant-theoretic characterization of
$\Delta$ with some Schubert calculus and geometric invariant theory
\cite{klyachko04,daftuarhayden04,klyachko06}; see also
\cite{christandldorankousidisetal12} for a different approach relying
solely on symplectic geometry. In the case of $N$ qubits, a complete
solution has been obtained in \cite{higuchisudberyszulc03}.
More generally, the way in which properties of the global state manifest in local correlations has also been studied in the literature \cite{lindenpopescuwootters02,wuerflingeretal12}.

As remarked above, convexity and invariant-theoretic descriptions can be established for arbitrary $G$-stable irreducible subvarieties $X \subseteq \PP(\mathcal H)$, not only for orbit closures \cite{brion87}. For instance, we may choose $X$ to be a \emph{Segr\'{e} variety}, i.e., the set of states which are biseparable with respect to a fixed bipartition of the $N$ subsystems (see \autoref{sec:genuine} below for an example).

\subsubsection{Experimental Noise}
\label{noise}

A quantum state prepared in the laboratory will always be a mixed
state $\rho$ and it is a priori unclear what statements can be
inferred about its entanglement from its local eigenvalues. Here, we
give two slightly different ways for leveraging the results discussed
for pure states to the more realistic scenario of small noise.

For both approaches, we will assume that a lower bound $p$ on the purity $\tr\rho^2$ is
available.  One natural way of obtaining such an estimate is the
well-known swap test \cite{buhrmanclevewatrousetal01} which directly
estimates $\tr\rho^2$ using two-particle measurements on two copies of
$\rho$.
We sketch an alternative procedure which may be simpler to implement
in some platforms. It is rigorous up to an assumption on the
prevailing noise mechanism: namely that it does not increase purity.
This does hold, e.g., for dephasing and depolarizing noise---two
models applicable to the majority of experiments.  Now suppose that
$\rho$ has been prepared by acting on an initial product state with a
quantum operation $\Lambda$ approximating an entangling unitary gate
$U$ (e.g., a spin squeezing operation).  Let $\Lambda'$ be a channel
approximating $U^{-1}$, and $\rho' = \Lambda'(\rho)$. Under said
assumption, we have that $\tr \rho'^2 \leq \tr \rho^2$, i.e., the
purity has decreased under the noisy ``disentangling operation''
$\Lambda'$. But $\rho'$ is still approximately a product, so that the
lower bound \cite{audenaert07}
\begin{equation*}
  \tr \rho^2
  \geq \tr \rho'^2
  \geq \sum_{k=1}^N \norm{\vec\lambda^{(k)}(\rho')}_2^2 - (N-1) = p
\end{equation*}
for the global purity in terms of the local eigenvalue spectra is
likely not to be too lose (it is tight for product states).
The following bounds will produce non-vacuous results only if the
purity $p$ exceeds $1/2$, which will from now on be assumed.

\bigskip

The first way of dealing with noise is to realize that there is a pure
state $\ket\psi$ with fidelity $\braket{\psi | \rho | \psi} \geq p$ whose
vector of local eigenvalues differs from the vector of local
eigenvalues of $\rho$ by at most $N \left(1-\sqrt{2p-1}\right)$.

To see this, expand the mixed state $\rho$ in its eigenbasis
$\rho=\sum_i r_i \ketbra {\psi_i} {\psi_i}$ with eigenvalues ordered non-increasingly, $r_i \geq r_{i+1}$.
Our assumption $p > 1/2$ on the purity implies immediately that the maximal eigenvalue is also lower-bounded by $1/2$,
\begin{equation}
  \label{r1vspurity}
	r_1 = r_1 \sum_i r_i \geq \sum_i r_i^2 = \tr \rho^2 = p > \frac 1 2.
\end{equation}
On the other hand, using that $r_j \leq \sum_{i>1} r_i = 1 - r_1$ for all $j > 1$, we find that
\begin{equation*}
  p
  = \tr \rho^2
  = r_1^2 + \sum_{i>1} r_i^2
  \leq r_1^2 + (1 - r_1)^2
  = 1 - 2 r_1 (1 - r_1).
\end{equation*}
We solve the quadratic relation and obtain two possible solutions,
\begin{equation*}
  r_1 \leq \frac 1 2 \left( 1 - \sqrt{2p - 1} \right),
  \qquad
  r_1 \geq \frac 1 2 \left( 1 + \sqrt{2p - 1} \right).
\end{equation*}
Only the right-hand side solution is compatible with $r_1 > 1/2$.
From now on we will employ the convention that
$\norm{-}_p$, when applied to matrices, denotes the (Schatten)
$p$-norm (in particular: $\norm{-}_1$ is the trace norm).
Set $\ket\psi := \ket{\psi_1}$. Then clearly, $\braket{\psi | \rho | \psi} = r_1 \geq p$ by \eqref{r1vspurity}.
On the other hand, by using a version of Weyl's perturbation theorem for
the $1$-norm \cite[(11.46)]{nielsenchuang04}, the vector of local eigenvalues only changes by
\begin{equation}
  \label{weyl perturbation}
  \sum_{k=1}^N \norm{\vec\lambda^{(k)}(\rho) - \vec\lambda^{(k)}(\ket\psi)}_1
  \leq N \norm{\rho - \ketbra \psi \psi}_1
	=
	2 N (1 - r_1)
	\leq
	N \left(1-\sqrt{2p-1}\right).
  %\leq 2 N \sqrt{1 - \braket{\psi | \rho | \psi}}
  %(= 2 N \sqrt{1 - p})
\end{equation}
For the case of $N$ qubits, a further improvement can be made. Here,
we are only concerned with the largest eigenvalue in each system. Due
to normalization $\lambda^{(k)}_{\max}+\lambda^{(k)}_{\min}=1$, every
deviation in the maximum eigenvalue has to be accompanied by
a deviation of equal magnitude of the minimum eigenvalue.
Thus, the 1-norm difference of the maximum eigenvalues is exactly half
the quantity estimated above:
\begin{equation}\label{eqn:qubit perturbation}
	\sum_{k}
	\abs{\lambda^{(k)}_{\max}(\rho) - \lambda^{(k)}_{\max}(\ket\psi)}
	=
	\sum_{k}
	\frac12
	\norm{\vec\lambda^{(k)}(\rho) - \vec\lambda^{(k)}(\ket\psi)}_1
	\leq
	\frac N2\left(1-\sqrt{2p -1}\right).
\end{equation}
For small noise, $p = 1 - \varepsilon \approx 1$, the right-hand side is given by $N\varepsilon/2$ to first order in $\varepsilon.$

We now illustrate this approach with a numerical example. Suppose that $\rho$ is an experimentally prepared quantum state of four qubits with purity $p = 0.9$. Then by the above there exists a pure state $\ket\psi$ with fidelity $\braket{\psi | \rho | \psi} \geq 0.9$ for which \eqref{eqn:qubit perturbation} reads
\begin{equation}
  \label{concrete bound}
  \sum_{k} \abs{\lambda^{(k)}_{\max}(\rho) - \lambda^{(k)}_{\max}(\ket\psi)} \leq \frac 4 2 \left( 1 - \sqrt{2 p - 1} \right) \approx 0.21.
\end{equation}
At this resolution, the differences between the various four-qubit entanglement polytopes are already well visible.
For concreteness, suppose that we would like to use the inequality
\begin{equation*}
	\lambda^{(1)}_{\max}(\ket\psi) + \lambda^{(2)}_{\max}(\ket\psi) + \lambda^{(3)}_{\max}(\ket\psi) + \lambda^{(4)}_{\max}(\ket\psi) < 3
\end{equation*}
to deduce that $\ket\psi$ is not entangled of W-type (compare main text).
For this, it would suffice by \eqref{concrete bound} to verify that the single-particle eigenvalues of the experimentally realized state $\rho$ satisfy the relation
\begin{equation*}
  \lambda^{(1)}_{\max}(\rho) + \lambda^{(2)}_{\max}(\rho) + \lambda^{(3)}_{\max}(\rho) + \lambda^{(4)}_{\max}(\rho) < 3 - 0.21 = 2.79.
\end{equation*}
For comparison, for the symmetric Dicke state
$\ket{D^{(2)}_4} = \frac 1 {\sqrt 6} \left( \ket{0011} + \ket{0101} + \ket{0110} + \ket{1001} + \ket{1010} + \ket{1100} \right)$
the left-hand side of the inequality is equal to $2$.

% (NB: The fidelity $\braket{D^{(4)}_2 | \rho | D^{(4)}_2} \approx 0.89$ reported in \cite{krischek11} implies that the left-hand side would be no larger than $2 + 4 \times \sqrt{0.11} \approx 3.32665$, which is not good enough.)

\bigskip

The second, alternative approach for treating noise consists in realizing that there exists a function
$\delta(p)$ such that if the local eigenvalues are more than
$\delta(p)$ away from an entanglement polytope $\Delta_{\mathcal C}$,
then $\rho$ cannot be written as a convex combination of pure states
$\ket{\phi_i}\in \overline{\mathcal C}$. More precisely, define the distance
between an experimentally obtained collection of eigenvalues $\lambda(\rho)$ and an
entanglement polytope $\Delta_{\mathcal C}$ by
\begin{equation*}
  d\left(\lambda(\rho), \Delta_{\mathcal C}\right)
  :=
  \min_{\mu \in \Delta_{\mathcal C}}
  \,
  \sum_{k=1}^N\norm{\vec \lambda^{(k)}(\rho) - \vec\mu^{(k)}}_1.
\end{equation*}
Then the desired function $\delta(p)$ is characterized by the property that
\begin{equation*}
  d\left(\lambda(\rho),\Delta_{\mathcal C}\right)> \delta(p)
  \>
  \Longrightarrow
  \>
  \rho \not\in \conv \big\{ \ketbra \phi \phi \,:\, \phi \in
  \overline{\mathcal C}\big\}.
\end{equation*}
A simple estimate can be derived along the lines of the previous paragraph. Indeed, assume that
$d\left(\lambda(\rho),\Delta_{\mathcal C}\right) > \delta(p)$ for some
value $\delta(p)$ that is yet to be determined.
Let
$\ket\psi = \ket{\psi_1}$ be as above, and let
$\ket\phi\in{\mathcal C}$ be arbitrary. Then
\begin{equation*}
  \norm{\ketbra \psi \psi - \ketbra \phi \phi}_1
  \geq
  \norm{\ketbra \phi \phi - \rho}_1
  -
  \norm{\rho - \ketbra \psi \psi}_1
  >
  \frac{\delta(p)}{N} - \left( 1 - \sqrt{2p - 1} \right) =: \epsilon(p),
\end{equation*}
where the first step is the triangle inequality, and the second step
mimics \eqref{weyl perturbation}. Borrowing another estimate from
\cite{fuchsgraaf99}, this implies that
$
  \abs{\langle\psi|\phi\rangle}^2 < 1 - \frac 1 4 \epsilon(p)^2
$
for all $\ket{\phi} \in \mathcal C$.
Now assume for the sake of reaching a contradiction that $\rho =
\sum_i r_i \ketbra {\phi_i} {\phi_i}$ can in fact be written as a convex
combination of pure states $\ket{\phi_i}$ from $\mathcal C$.
\begin{equation*}
  \braket{\psi | \rho | \psi}
  = \sum_i r_i \, \abs{\langle\psi|\phi_i\rangle}^2
 < 1-\frac14 \epsilon(p)^2.
\end{equation*}
On the other hand, $\braket{\psi | \rho | \psi} \geq p$, which is a contradiction
whenever $p \geq 1-\frac14 \epsilon(p)^2$.
A short calculations proves that this is certainly the case if we choose
$\delta(p) := N \left( 1 + 2 \sqrt{1-p} - \sqrt{2p - 1} \right)$.
%$\delta(p) := 4 N \sqrt{1-p}.$
%We suspect that tighter bounds can be found.

\subsection{Linear Entropy of Entanglement and its Distillation}
\label{entropy of entanglement}

\subsubsection{Entanglement Polytopes and the Origin}
\label{origin}

Pure quantum states $\rho = \ketbra \psi \psi$, whose one-body reduced density matrices are maximally mixed, i.e., $\rho^{(k)} \propto \id$, play a special role in entanglement theory. First, any such state maximizes the total uncertainty $\sum_j (\Delta R_j)^2$ of any orthonormal basis of local observables $R_j$ \cite{klyachko07}.
% Indeed, we may assume without loss of generality that the $R_j$ have trace zero. Then, as was shown in \cite{klyachko07}
% \begin{equation*}
%   \sum_j (\Delta R_j)^2 =
%   \sum_j \left( \tr(\rho \, R_j^2) - \tr(\rho \, R_j)^2 \right) =
%   \tr(\rho \, \sum_j R_j^2) - \sum_j \tr(\rho \, R_j)^2 =
%   c - \sum_j \tr(\rho \, R_j)^2,
% \end{equation*}
% since $\sum_j R_j^2$ is the Casimir operator which acts by multiplication with some scalar $c$ on the irreducible $G$-representation $\mathcal H$. It follows that the total uncertainty is maximal if and only if the expectation values of any local observable is zero, that is, if and only if all reduced density matrices are proportional to the identity.
Second, any entanglement monotone $M(\ketbra \psi \psi) := \abs{P(\ket \psi)}$ defined in terms of a $G$-invariant homogeneous polynomial $P$ attains on $\rho$ its maximal value over all the states in the same entanglement class \cite{klyachko07,verstraetedehaenedemoor03}. To see this, observe that
\begin{equation*}
  M(\rho)
  = \abs{P(\ket \psi)}
  = \abs{P(g \ket \psi)}
  = \norm{g \ket \psi}^n \, \abs{P\left(\frac{g \ket \psi}{\norm{g \ket \psi}} \right)}
  \geq \abs{P\left(\frac{g \ket \psi}{\norm{g \ket \psi}}\right)}
  = M(g \cdot \rho)
\end{equation*}
where $n$ is the degree of $P$; the inequality follows from a result by Kempf and Ness which states that $\norm{g \ket \psi} \geq \norm{\ket \psi}$ for all $g \in G$ if the $\rho^{(k)}$ are maximally mixed \cite{kempfness79}.

The point in the entanglement polytopes corresponding to quantum states with maximally mixed one-body reduced density matrices will be called the \emph{origin} and denoted by $O$; it satisfies $\vec\lambda^{(k)} \propto (1,\ldots,1)$ for all $k$.
Clearly, $O$ cannot be expressed as a proper convex combination of any two other points in the entanglement polytope, for we have arranged each vector of eigenvalues in weakly decreasing order. Therefore, \autoref{main thm} implies that the origin is contained in the entanglement polytope if and only if there exists a covariant with weights $\vec\mu^{(k)} \propto (1,\ldots,1)$. Such a covariant is of course nothing but a $G$-invariant homogeneous polynomial; hence,
\begin{equation*}
  O \in \Delta_{\mathcal C}
  \,\Longleftrightarrow\,
  \text{$\exists$ a $G$-invariant homogeneous polynomial which is non-zero at $\rho$}.
\end{equation*}
In particular, any entanglement monotone defined via polynomial
invariants necessarily vanishes on those quantum states whose
entanglement polytopes do not include the origin (these states are
also called \emph{unstable} in geometric invariant theory
\cite{mumfordfogartykirwan94}; their complement being the
\emph{semi-stable} states).
This observation has lead to the suggestion that unstable states
should be considered ``unentangled''
\cite{klyachko07}
or ``not genuinely multipartite entangled'' \cite{osterloh_constructing_2005,osterlohsiewert06}, even though
they might be entangled according to the standard notion of
entanglement that we have adopted in this work.

\subsubsection{Linear Entropy of Entanglement}

The geometric picture provided by entanglement polytopes suggests another way of quantifying entanglement, which assigns non-zero values to all entangled states, namely the Euclidean distance of the collection of local eigenvalues. This quantity is precisely equal to the $\ell_2$-norm distance by which the reduced density matrices differ from being maximally mixed. Moreover, it is by an affine transformation related to the multi-particle version of the \emph{(linear) entropy of entanglement} \cite{zurekhabibpaz93,furuyanemespellegrino98,barnumknillortizetal04} given by
\begin{equation*}
  E(\rho)
  = E(\vec\lambda^{(1)}, \ldots \vec\lambda^{(N)})
  = 1 - \frac{1}{N}\sum_{k=1}^N \norm{\vec\lambda^{(k)}}^2_2
  = 1 - \frac{1}{N}\sum_{k=1}^N \norm{\rho^{(k)}}^2_2.
\end{equation*}
Indeed,
\begin{equation}
  \label{euclid vs entropy}
  \norm{(\vec\lambda^{(1)}, \ldots, \vec\lambda^{(N)}) - O}^2_2
  = \sum_{k=1}^N \norm{\rho^{(k)} - \frac \id {d_k}}^2_2
  = \sum_{k=1}^N \norm{\rho^{(k)}}^2_2 - \frac 1 {d_k}
  = \left(N - \sum_{k=1}^N \frac 1 {d_k} \right) - N E(\rho).
\end{equation}
The linear entropy of entanglement admits an operational interpretation \cite{boixomonras08}; in the case of multiple qubits, it reduces to the Meyer--Wallach measure of entanglement, and to the concurrence in the case of two qubits \cite{meyerwallach02,brennen03,hillwootters97}.

The following result then follows immediately from convexity.
It naturally generalizes the properties satisfied by states with maximally mixed one-body reduced density matrix.

\begin{prp}
  Any entanglement polytope $\Delta_{\mathcal C}$ contains a unique point of minimal Euclidean distance to the origin $O$.
  The corresponding quantum states $\rho$ maximize the linear entropy of entanglement $E(\rho)$ over all states in $\overline{\mathcal C}$.
\end{prp}

The maximal entropy of entanglement that can be obtained from states in the closure of the entanglement class $\mathcal C$,
\begin{equation*}
  E(\Delta_{\mathcal C})
  = \max \{ E(\rho) : \rho \in \overline{\mathcal C} \}
  = \max \{ E(\vec\lambda^{(1)}, \ldots \vec\lambda^{(N)}) : (\vec\lambda^{(1)}, \ldots \vec\lambda^{(N)}) \in \Delta_{\mathcal C} \}
\end{equation*}
is in view of \eqref{euclid vs entropy} a simple function of the Euclidean distance of the entanglement polytope to the origin and can thus be computed easily.

\subsubsection{Entanglement Distillation}
\label{distillation}

Given the linear entropy of entanglement $E(\rho)$ as a means of quantifying multi-particle entanglement, it is natural to ask for a corresponding distillation procedure, i.e., a protocol for transforming the quantum state $\rho$ by SLOCC operations to a state with maximal linear entropy of entanglement. This might only be possible asymptotically, as the maximum might be contained in the boundary of $\overline{\mathcal C}$.

The function $E(\rho)$ is smooth and can therefore be maximized locally by following its gradient flow in $\mathcal C = G \cdot \rho$. By \eqref{moment map equation},
\begin{equation}
  \label{differential}
  \restrict{dE}{\rho}
  = - \frac{2}{N} d\braket{\Phi, (\rho^{(1)}, \ldots, \rho^{(N)})}\restrict{}{\rho}
  = - \frac{2}{N} \omega_{\text{FS}}(\widehat{i X}_\rho, -)
  = - \frac{2}{N} \tr(\widehat X_\rho -),
\end{equation}
where we denote by $\widehat X_\rho$ the tangent vector generated at $\rho$ by the infinitesimal action of the collection of reduced density matrices.
The last identity holds since $\PP(\mathcal H)$ is a K\"ahler manifold. Therefore, the gradient of $E$ is at every point $\rho = \ketbra \psi \psi$ equal to
\begin{equation}
  \label{gradient}
  \grad E(\rho)
  = - \frac 2 N \hat X_\rho
  = - \frac 2 N \restrict{\frac d {dt}}{t=0} \frac {(e^{t \rho^{(1)}} \otimes \ldots \otimes e^{t \rho^{(N)}}) \ketbra \psi \psi (e^{t \rho^{(1)}} \otimes \ldots \otimes e^{t \rho^{(N)}})} {\norm{(e^{t \rho^{(1)}} \otimes \ldots \otimes e^{t \rho^{(N)}}) \ket\psi}^2},
\end{equation}
which factors over the action \eqref{SL action} of the Lie algebra of $G$. Therefore, any solution $\rho_t$ to the gradient flow equation remains in the entanglement class of the initial value $\rho_0 = \rho$ at all times $t > 0$.

In mathematical terms, the Euclidean distance to the origin and hence $E(\rho)$ are closely related to the norm square of the moment map (cf.~\autoref{covariants}), and the latter is a minimally degenerate Morse function in the sense of Kirwan \cite{kirwan84}. The corresponding analog of \eqref{gradient} had already been noticed in \cite{kirwan84}. Moreover, it was established that while the solution $\rho_t$ of the gradient flow equation might not necessarily converge to a unique limit point, $E$ is sufficiently well-behaved so that $E(\rho_t)$ always converges to the global maximum as $t \rightarrow \infty$. We summarize:

\begin{thm}
  \label{distillation thm}
  By following the gradient flow of $E$, which at any point $\rho$ is given by the infinitesimal action of the reduced density matrices \eqref{gradient}, the global maximum of the linear entropy of entanglement in the closure of the entanglement class of $\rho$ is reached (possibly asymptotically).
\end{thm}

In practice, the gradient flow in \autoref{distillation thm} would need to be implemented with finite time steps.
That is, after preparing the quantum state $\rho$ one measures its one-body reduced density matrices, re-prepares and performs local POVM measurements with Kraus operators $S^{(k)} \propto e^{t \rho^{(k)}}$, $F^{(k)} = \sqrt{\id - (S^{(k)})^\dagger S^{(k)}}$ for sufficiently small but finite $t > 0$ (cf.~\eqref{SL POVM}). If the outcomes of these measurements are $S^{(1)}$, \ldots, $S^{(N)}$ then entanglement has been distilled. By successively repeating this procedure and concatenating the SLOCC operations, one asymptotically arrives at a quantum state with maximal entropy of entanglement.
Notably, this method of entanglement distillation only requires local tomography and works on a single copy of the state at a time.

From a theoretical perspective, the limit point $\lim_{t \rightarrow \infty} \rho_t$ of the gradient flow can also be seen as a normal form of the state $\rho$ in its entanglement class (in case this limit exists). This is the point of view taken in \cite{verstraetedehaenedemoor03}, where a similar numerical algorithm has been given for the special case where $O$ is contained in the entanglement polytope.

\subsection{Examples}
\label{examples}

By associating to every entanglement class its entanglement polytopes, we have obtained a finite yet systematic classification of multi-particle entanglement. In this section, we will illustrate their computation and application by a series of examples for systems of several qubits.

Mathematically, the covariants of a multi-qubit system are in one-to-one correspondence with the covariants of binary multilinear forms, whose study is a prominent topic in classical invariant theory \cite{olver99,luque07}.
Before we proceed, we introduce some notational simplifications. It will be convenient to represent the local eigenvalues of a system of $N$ qubits by the tuple $(\lambda^{(1)}_{\max}, \ldots, \lambda^{(N)}_{\max}) \in [1/2,1]^N$ of \emph{maximal} local eigenvalues. This is of course without loss of information, since the sum of the two eigenvalues of any qubit density matrix, which is a $2 \times 2$ positive semidefinite Hermitian matrix of trace one, is equal to unity.
Similarly, we may label the weights of a covariant by the tuple $(\mu^{(1)}_{\max}, \ldots, \mu^{(N)}_{\max}) \in \{\lceil n/2 \rceil, \ldots, n \}^N$, since the sum of the components of each weight $\vec\mu^{(k)}$ has been fixed to be equal to the degree $n$ of the covariant.

\subsubsection{Three Qubits}
\label{three qubits example}

In the case of three qubits, with Hilbert space $\mathcal H = \CC^2 \otimes \CC^2 \otimes \CC^2$, there are six distinct entanglement classes \cite{durvidalcirac00}: The classes of the GHZ state \cite{greenbergerhornezeilinger89} and of the W state \cite{durvidalcirac00},
\begin{equation*}
  \ket\GHZ = \frac{1}{\sqrt{2}}(\ket{\uparrow\uparrow\uparrow} + \ket{\downarrow\downarrow\downarrow}), \quad
  \ket W = \frac{1}{\sqrt{3}}(\ket{\uparrow\uparrow\downarrow} +\ket{\uparrow\downarrow\uparrow} + \ket{\downarrow\uparrow\uparrow}),
\end{equation*}
three classes that correspond to EPR states shared between any two of the three subsystems,
\begin{equation*}
  \ket{B1} = \frac 1 {\sqrt 2} (\ket{\uparrow\uparrow\downarrow} - \ket{\uparrow\downarrow\uparrow}), \quad
  \ket{B2} = \frac 1 {\sqrt 2} (\ket{\uparrow\uparrow\downarrow} - \ket{\downarrow\uparrow\uparrow}), \quad
  \ket{B3} = \frac 1 {\sqrt 2} (\ket{\uparrow\downarrow\uparrow} - \ket{\downarrow\uparrow\uparrow}),
\end{equation*}
and the separable class represented by $\ket\SEP = \ket{\uparrow\uparrow\uparrow}$.

We shall now compute the corresponding entanglement polytopes by
following the general method of covariants described in
\autoref{properties and computation}. Using techniques crafted towards
the special situation of three qubits, the same polytopes have already
been computed in \cite{han_compatible_2004,sawickiwalterkus12}; the
corresponding quantum marginal problem, which as we have explained
amounts to computing the maximal entanglement polytope, has been
solved in \cite{higuchisudberyszulc03}. A minimal set of generators of
the covariants of a three-qubit system (in fact, of the equivalent
question for binary three-linear forms), has been determined in late
19th century invariant theory \cite{lepaige81}: There are six
generators, and we have summarized their properties in \autoref{three
qubit covariants}. By \autoref{main cor}, computing the entanglement
polytopes is now a mechanical task: for any quantum state representing
the entanglement class, we merely need to collect those covariants
which do not vanish on the state, and take the convex hull of their
normalized weights.

The resulting entanglement polytopes are illustrated in \autoref{three qubit polytopes}. They are in one-to-one correspondence to the six entanglement classes described above; that is, in this particular case there is no coarse-graining.
As explained before, one polytope is contained in the other if quantum states in the former class can be approximated arbitrarily well by states in the latter class. In this case, this is also a necessary condition, since there is no coarse-graining. Since the GHZ-class polytope is maximal, it follows that all states can be approximated arbitrarily well by states of GHZ type. In mathematical terms, the GHZ class is dense; this is of course well-known \cite{durvidalcirac00}. Similarly, the polytope of the W class (upper pyramid) contains all entanglement polytopes \emph{except} the GHZ one, so that by states in the W class one can approximate all states except those of GHZ class.

We now illustrate our method of entanglement witnessing:
\begin{itemize}
\item If the point $(\lambda^{(1)}_{\max},\lambda^{(2)}_{\max},\lambda^{(3)}_{\max})$ corresponding to the collection of local eigenvalues is contained in the lower part of the GHZ entanglement polytope (\autoref{three qubit polytopes}, (A)),
  \begin{equation*}
    \lambda^{(1)}_{\max} + \lambda^{(2)}_{\max} + \lambda^{(3)}_{\max} < 2,
  \end{equation*}
  then it is by \eqref{appendix criterion} not contained in any other entanglement polytope, and therefore the quantum state at hand must be entangled of GHZ type.

\item More generally, if the point is \emph{not} contained in any of the polytopes corresponding to an EPR state shared between two of the three particles (\autoref{three qubit polytopes}, (c), which includes (d)), i.e., if
  \begin{equation*}
    \lambda^{(k)}_{\max} < 1
    \quad
    (\forall k=1,2,3),
  \end{equation*}
  then by \eqref{appendix criterion} the quantum state at hand must be entangled of either GHZ or W classes.
  These classes of states are the ones that possess genuine three-qubit entanglement.
\end{itemize}

As a final example, we consider the quantum state $\rho = \ketbra \psi
\psi$, where
\begin{equation*}
  \ket\psi = \frac 1 {\sqrt 7} \left( \ket{\uparrow\uparrow\uparrow} + \ket{\uparrow\uparrow\downarrow} + \ket{\downarrow\uparrow\uparrow} + 2 \ket{\downarrow\downarrow\uparrow} \right).
\end{equation*}
It is easy to verify by the method of covariants that the entanglement
polytope of $\rho$ is full-dimensional, and by the above
classification it follows that $\rho$ is of GHZ type. However, its
collection of local eigenvalues,
$(\lambda^{(1)}_{\max},\lambda^{(2)}_{\max},\lambda^{(3)}_{\max})
\approx (0.76, 0.79, 0.88)$, is contained in the interior of the upper
pyramid. As we have just discussed, the entanglement criterion
\eqref{appendix criterion} in this case only allows us to conclude that $\rho$
is of either GHZ or W type. By using the entanglement distillation
procedure described in \autoref{distillation} we can however transform
$\rho$ by into another state whose local eigenvalues are arbitrarily
closed to the origin $O = (0.5,0.5,0.5)$ (\autoref{three qubit
distillation}). In this way, we arrive at quantum states which are
both more entangled and for which our entanglement criterion
\eqref{appendix criterion} is maximally informative.

\subsubsection{Four Qubits}
\label{four qubits example}

The case of three qubits was rather special, since there the
entanglement polytopes represent faithfully the hierarchy of
closures of the corresponding entanglement classes---no coarse
graining takes place.
In contrast, for four qubits the situation is the
generic one: There are infinitely many entanglement classes.
According to the classification of
\cite{verstraetedehaenedemooretal02}, they can be partitioned into
nine families with up to four complex continuous parameters each.
Neither the families themselves, nor the complex parameters within these
families are directly experimentally accessible.

Here, we sketch a proof showing that the polytopal view strikes an
attractive balance between reducing the complexity sufficiently to
allow for simple experimental criteria on the one hand, and preserving
a non-trivial structure on the other hand. Indeed, following the
tradition established in \cite{durvidalcirac00} and continued in
\cite{verstraetedehaenedemooretal02}, we may state that through
the polytopal lense, \emph{four qubits can be entangled in seven different ways}.

We have determined all entanglement polytopes of four qubits using
the general method (\autoref{properties and computation})
applied to a minimal generating set of 170
covariants found in \cite{briandluquethibon03}. More precisely,
for every family in \cite{verstraetedehaenedemooretal02}, we have
analytically computed the covariants using a computer algebra system.
Deciding whether a normalized weight $\frac1{n_j}\vec\mu_j^{(k)}$ is included
in $\Delta_{\mathcal C}$ then amounts to solving the explicit polynomial
equation $P_j(\rho)\neq 0$. The algebra system can readily perform
such calculations analytically.

The polytope $\Delta$ formed by the marginal eigenvalues of all
possible pure states is the convex hull of $12$ vertices. These can
be easily described as follows: One vertex,
$(\lambda_{\max}^{(1)}, \dots, \lambda_{\max}^{(4)})=(1,1,1,1)$,
corresponds to product states; six vertices
\begin{equation*}
	(0.5, 0.5, 1, 1), \>
	(0.5, 1, 0.5, 1), \>
	(0.5, 1, 1, 0.5), \>
	(1, 0.5, 0.5, 1), \>
	(1, 0.5, 1, 0.5), \>
	(1, 1, 0.5, 0.5)
\end{equation*}
belong to two-partite entangled states; four vertices
\begin{equation*}
	(0.5, 0.5, 0.5, 1), \>
	(0.5, 0.5, 1, 0.5), \>
	(0.5, 1, 0.5, 0.5), \>
	(1, 0.5, 0.5, 0.5)
\end{equation*}
to three-partite entangled states, and one vertex $(0.5, 0.5, 0.5, 0.5)$
which is the image of a four-partite entangled state (not necessarily
\emph{genuinely} four-partite entangled---see \autoref{sec:genuine}).
This latter vertex is the origin as defined in \autoref{origin}.

As in the three-qubit case, there are several lower-dimensional
subpolytopes corresponding to bi-separable states. These are obtained
by embedding the entanglement polytopes found in the previous section
into the full four-qubit polytope in the obvious way. All
possibilities are listed in \autoref{four qubit lower dimensional}.

We turn to the genuinely four-body entangled states. There are
seven such polytopes, all full-dimensional. Their definitions and some of their
properties are listed in \autoref{four qubits table} and
\autoref{fig:four qubits}. The four-qubit $W$-state (or Dicke
state)
$
\ket{W_4} = \frac{1}{\sqrt{4}}(
\ket{\uparrow\uparrow\uparrow\downarrow} +
\ket{\uparrow\uparrow\downarrow\uparrow} +
\ket{\uparrow\downarrow\uparrow\uparrow} +
\ket{\downarrow\uparrow\uparrow\uparrow})
$
corresponds to polytope~5. As is the case for three qubits, the
polytope is an ``upper pyramid'', i.e., it is the intersection of the
full polytope $\Delta$ with the half-space defined by
\begin{equation}
  \label{four qubit W}
  \lambda^{(1)}_{\max} + \lambda^{(2)}_{\max} +
	\lambda^{(3)}_{\max} + \lambda^{(4)}_{\max} \geq 3.
\end{equation}
Again in analogy to the three-qubit case, we can take any violation of
\eqref{four qubit W} as an indication of ``high entanglement''.
One way to make this precise is to read off \autoref{four qubits table} that violations imply that the state $\ket\psi$ can be converted into
one with entropy of entanglement at least $0.45$ (which might be much
higher than $E(\ket\psi)$ as obtained from the measured data!).
%by a SLOCC protocoll
%As mentioned in the main text, we point out that it may be surprising
%that such strong conclusion can be drawn from single-party information
%alone.
The entanglement classes of the four-qubit GHZ state and of the cluster states \cite{briegelraussendorf01} are associated with the full polytope (number~7).

There is a numerical coincidence between our findings and the ones in
\cite{dokovicosterloh09}, where also seven non-biseparable
entanglement classes have been identified on four qubits. The two
classifications are, however, not
identical.  Indeed, \cite{dokovicosterloh09} is based purely on
\emph{invariants}, as opposed to the more general
\emph{covariant}-theoretic description of our polytopes. As mentioned
in \autoref{entropy of entanglement}, invariants cannot differentiate
unstable entangled states from product states. Hence, the 7 classes of
\cite{dokovicosterloh09} must all be semi-stable. However the same is
true only for our polytopes 4, 6 and 7. In this sense, the polytope
methods yields a finer classification for unstable vectors, whereas
\cite{dokovicosterloh09} provides a better resolution of the stable
case.

In summary, up to permutations, there are 12 entanglement polytopes for
four qubits, 7 of which belong to genuinely four-partite entangled
classes. The numbers increase to 41 and 22, respectively, if distinct
permutations are counted separately.

% TBD: better to list additional half space constraints?

\subsubsection{\texorpdfstring{$N$ Qubits and Genuine Multipartite Entanglement}{N Qubits and Genuine Multipartite Entanglement}}
\label{sec:genuine}

In the examples so far, bi-separable states mapped onto polytopes of
lower dimension. This is no longer true for $N\geq 6$ particles (for $N=6$, the
entanglement polytope associated with $\ket{GHZ_3}\otimes\ket{GHZ_3}$
is clearly 6-dimensional). However, it remains true that spectral
information alone can be used to witness genuine $N$-qubit
entanglement for any $N$.
A general theory of genuine entanglement detection from spectral
information will be presented elsewhere. Here, we merely give one
example valid for any number $N$ of qubits. The result is stated in
the language of \cite{guehne_multipartite_2005}; see also
\cite{horodecki_separability_2001,horodecki09,guhne_entanglement_2009,levi_hierarchy_2012,seevinck_sufficient_2001,huber_detection_2010}
and references therein.
Call a vector $\ket\psi$ \emph{producible using $k$-partite
entanglement} if it is of the form
$\ket\psi=\ket{\psi_1}\otimes\dots\otimes\ket{\psi_m}$, where every
$\ket{\psi_i}$ is contained in the tensor product of at most $k$
sites---otherwise, $\ket\psi$ is said to contain \emph{genuine $(k+1)$-partite entangled}. In particular, the states that contain genuine $N$-partite entanglement are precisely those which are not \emph{biseparable}, i.e., those which do not factorize with respect to any non-trivial (bi)partition of the subsystems.
The reader should not be confused by the fact that some authors use
the term ``genuine $N$-partite entanglement'' in a different sense;
c.f.\ \cite{osterloh_constructing_2005,osterlohsiewert06}, where semi-stability (in the sense of
\autoref{entropy of entanglement}) is additionally required for ``genuine''
entanglement.

Recall first that for a system of $N$ qubits the solution of the quantum marginal problem (i.e., the maximal entanglement polytope) is given by the inequalities
\begin{equation}
  \label{higuchi}
  \lambda^{(i)}_{\min} \leq \sum_{i \neq j} \lambda^{(j)}_{\min}
\end{equation}
for the \emph{smallest} local eigenvalues $\lambda^{(i)}_{\min} = 1-\lambda^{(i)}_{\max} \in [0,0.5]$ \cite{higuchisudberyszulc03}. Therefore, the polytope for the class of states that factorize with respect to a given partition $A_1 \cup \ldots \cup A_m = \{1,\ldots,N\}$ is defined by the constraints
\begin{eqnarray}
  \label{higuchi for partition}
  \lambda_{\min}^{(i)} \leq \sum_{i \neq j \in A_l}
  \lambda^{(j)}_{\min} &\qquad&
  \forall\> i \in A_k,\>k=1,\ldots,m
\end{eqnarray}

The following lemma shows that information on the local eigenvalues can serve as a witness for genuine $k$-qubit entanglement (cf.~Fig.~3):

\begin{lem}
  Let $N \geq 3$ and $k \in \{\lceil N/2 \rceil,\ldots,N\}$. For every $\gamma\in(0,0.5\frac{N-1}{k-1}]$ the local eigenvalues
  \begin{equation*}
    (\lambda_{\min}^{(1)}, \dots, \lambda_{\min}^{(N)})
    =
    \left(\gamma \frac{k-1}{N-1},
    \frac{\gamma}{N-1}, \dots, \frac{\gamma}{N-1}\right)
  \end{equation*}
  can originate only from genuinely $k$-partite entangled (pure) quantum states of $N$ qubits.
  Conversely, if $k\neq N-1$ then there exists a realization in terms of a state producible using $k$-partite entanglement.
\end{lem}
\begin{proof}
  Consider any bipartition $A_1 \cup A_2 = \{1, \ldots, N\}$.
  Without loss of generality, assume that $1 \in A_1$.
  Then \eqref{higuchi for partition} for $i=k=1$ reads
  \begin{equation}
    \label{higuchi for the lemma}
    \gamma\frac{k-1}{N-1}
    \leq
    \sum_{1 \neq j \in A_1} \gamma\frac{1}{N-1}
    =\gamma\frac{\abs{A_1}-1}{N-1}
    \qquad
    \Leftrightarrow
    \qquad
    \abs{A_1} \geq k.
  \end{equation}
  This immediately proves that any $\ket\psi$ with the advertised
  local spectra is genuinely $k$-partite entangled. For if such a $\ket\psi$ factorizes with respect to some partition, then by \eqref{higuchi for the lemma} at least one factor must comprise at least $k$ sites.

  To show that the local eigenvalues can indeed be realized using $k$-partite entanglement, consider the bipartition $A_1 = \{1,\ldots,k\}$, $A_2 = \{k+1,\ldots,N\}$. Then \eqref{higuchi for the lemma} is satisfied, and so are all other inequalities in \eqref{higuchi for partition} for $A_1$. The constraints for $A_2$ are satisfied if and only if $\abs{A_2} \neq 1 \, \Leftrightarrow \, k\neq N-1$, which is true by assumption.
\end{proof}

In other words, certain correlations between the one-particle reduced density matrices can only be explained by the presence of genuine $k$-partite entanglement.
This result holds in fact for all choices of local dimensions, since we can always reduce to the qubit situation by considering local eigenvalue spectra of rank at most two.

\subsubsection{Bosons and Fermions}

The theory of entanglement polytopes is readily adapted to the case of indistinguishable particles. Here, the Hilbert space is no longer a tensor product, but the symmetric or antisymmetric subspace in $(\CC^d)^{\otimes N}$, where $d$ is the local dimension. Since the particles are indistinguishable, all one-body reduced density matrices are equal. Therefore, the entanglement polytope of an entanglement class $\mathcal C$ is in the case of indistinguishable particles defined to be
\begin{equation}
	\label{entanglement polytope bosons fermions}
  \Delta_{\mathcal C} =
  \left\{
    \vec\lambda^{(1)}
    \,:\,
    \vec\lambda^{(1)} = \vec\lambda(\rho^{(1)}),\,
    \rho \in \overline{\mathcal C}
  \right\}.
\end{equation}
The reduced density matrices can be diagonalized by the coherent action of the group $K = \SU(d)$, and the corresponding group of SLOCC operations is $G = \SL(d)$ \cite{mathonetkrinsgodefroidetal10}. Using these definitions, the theory can be developed in precisely the same way as above.

\bigskip

We shall now illustrate this for the case of a system of $N$ indistinguishable two-level systems with bosonic statistics, with Hilbert space $\mathcal H = \Sym^N(\CC^2) \subseteq (\CC^2)^{\otimes N}$.
As in the preceding examples, we may represent the vector of local eigenvalues of the one-body reduced density matrix by its maximum. Thus we may think of an entanglement polytope $\Delta_{\mathcal C}$ as a line segment or interval in $[0.5,1]$.

The entanglement polytopes can be computed as above using covariants, which in this case correspond formally to the covariants of binary forms in mathematics; these are again well-studied and explicitly known for small $N$ \cite{kungrota84,olver99}. As a convenient shortcut, we will however use a geometric argument which immediately gives the entanglement polytopes for arbitrary $N$:

\begin{lem}
  \label{boson lemma}
  For a system of $N$ bosonic qubits, the entanglement polytopes are given by intervals $[\gamma,1]$ with
  \begin{equation*}
    \gamma \in \left\{ \frac 1 2 \right\} \cup \left\{ \frac {N - k} N : k = 0, 1, \ldots, \lfloor N/2 \rfloor \right\}.
  \end{equation*}
\end{lem}
\begin{proof}
  Since the unentangled states are contained in the closure of every entanglement class, every entanglement polytope $\Delta_{\mathcal C}$ contains the point $1$. The other end point of the interval, which we denote by $\gamma \geq 0.5$, is by \eqref{euclid vs entropy} directly related to the maximal entropy of entanglement $E(\Delta_{\mathcal C})$, and we need to determine the possible values of $\gamma$. First, we observe that since the generalized GHZ state $\frac 1 {\sqrt 2} ( {\ket \uparrow}^{\otimes N} + {\ket \downarrow}^{\otimes N} )$ is locally maximally mixed, $\gamma = 0.5$ can always be achieved for arbitrary $N$. Let us now suppose that $\gamma > 0.5$, i.e., the spectrum of the one-body reduced density matrix is non-degenerate and hence contained in the interior of the set $D_\downarrow$. Let $\rho = \ketbra \psi \psi$ be a quantum state whose one-body reduced density matrix is a diagonal matrix with first diagonal entry $\gamma$. By \eqref{moment polytope}, $\gamma$ can only be a vertex of the entanglement polytope if the range of the differential of $\Phi$ at $\rho$ is orthogonal to the space of diagonal observables spanned by $\sigma_z$ (otherwise, $\gamma$ could be changed infinitesimally, hence would not be a vertex). By \eqref{moment map equation}, this is the case if and only if $\widehat{i \sigma_z}_\rho$ vanishes---in other words, $\rho$ is fixed by the infinitesimal action of $i\sigma_z$ (and hence of any diagonal local operator) and thus $\ket\psi$ is an eigenvector of the $\sigma_z$-operator. Such states are known as occupation number basis states, or Dicke states \cite{dicke54}. It is easy to see that the one-body reduced density matrix of a Dicke state with $N_\uparrow$ spins pointing upward and $N_\downarrow = N - N_\uparrow$ spins pointing downward is equal to $\diag( N_\uparrow/N, N_\downarrow/N)$; therefore it is mapped into $D_\downarrow$ if $N_\uparrow \geq N_\downarrow$. Finally, we observe that since the one-body reduced density matrix is diagonal and the global state fixed by diagonal local operators, every Dicke state is a maximum of the entanglement distillation procedure from \autoref{distillation}. It follows that $\gamma$ is indeed minimal, hence a vertex of the entanglement polytope.
\end{proof}

By associating with each of the entanglement polytopes the set of corresponding quantum states, we obtain $\lceil N/2 \rceil + 1$ families of quantum states. Each family generically consists of infinitely many entanglement classes, except for the unentangled one ($\gamma=1$), and it can be characterized operationally by the maximal linear entropy of entanglement $E(\Delta_{\mathcal C})$ that can be achieved by states in the family.
Furthermore, it follows from the proof of \autoref{boson lemma} that the Dicke states for $\gamma > 0.5$ are (up to local unitaries) uniquely characterized by the property that they attain the maximal entropy of entanglement in their entanglement polytope. In contrast, there are generically infinitely many entanglement classes that contain quantum states which are locally maximally mixed (i.e., sent to the origin, corresponding to $\gamma = 0.5$); this follows because for $N \geq 4$ there is more than a single invariant. For example, the four qubit GHZ state and the Dicke state with $N_\uparrow = N_\downarrow = 2$ are both locally maximally mixed, but in a different entanglement class.

% Proof by using invariants and the GIT quotient.
% However, in the closure of each entanglement _class_ there is only a single such state (up to local unitaries)
% Amusingly, this is closely related to Kirwan's stratification from \cite{kirwan84}.

\bigskip

Our second example is a system of three fermions with local rank six, as described by the Hilbert space $\mathcal H = \bigwedge^3 \CC^6$.
In \cite{borlanddennis72}, Borland and Dennis proved that the set of local eigenvalues compatible with a global pure state is constrained by the inequalities
\begin{align}
	\label{borland dennis}
	\lambda_1 + \lambda_6 =
	\lambda_2 + \lambda_5 =
	\lambda_3 + \lambda_4 = 1,
	\quad
  \lambda_1 + \lambda_2 \leq \lambda_3 + 1.
\end{align}
Here we write $\lambda_k = \lambda^{(1)}_k$ for the $k$-th eigenvalue of the one-body reduced density matrix $\rho^{(1)}$, which---following usual conventions in quantum chemistry---is normalized to $\tr \rho^{(1)} = 3$.
This is the solution to the quantum marginal problem, which in this context is also known as the \emph{pure-state $N$-representability problem}.
Notably, \eqref{borland dennis} describes a three-dimensional convex polytope in six-dimensional Euclidean space (due to the three equality constraints). We may therefore without loss of information work in three-dimensional space by consider its projection onto the largest three eigenvalues, $(\lambda_1,\lambda_2,\lambda_3)$.

In the Borland--Dennis system there are four distinct entanglement classes \cite{levayvrana08}; see also \cite{schouten31,ehrenborg99}. They can represented by the following states:
\begin{equation}
	\label{borland dennis orbits}
\begin{aligned}
	\ket{\psi_A} &= \frac 1 {\sqrt 2} \left(
		\ket 1 \wedge \ket 2 \wedge \ket 3 +
		\ket 4 \wedge \ket 5 \wedge \ket 6
	\right), \\
	\ket{\psi_B} &= \frac 1 {\sqrt 3} \left(
		\ket 1 \wedge \ket 2 \wedge \ket 4 +
		\ket 1 \wedge \ket 3 \wedge \ket 5 +
		\ket 2 \wedge \ket 3 \wedge \ket 6
	\right), \\
	\ket{\psi_C} &= \frac 1 {\sqrt 2} \ket 1 \wedge \left( \ket 2 \wedge \ket 3 + \ket 4 \wedge \ket 5 \right), \\
	\ket{\psi_D} &= \ket 1 \wedge \ket 2 \wedge \ket 3.
\end{aligned}
\end{equation}
Observe that $\ket{\psi_C}$ is the antisymmetrization of a biseparable pure state, while the class of $\ket{\psi_D}$ is equal to set of Slater determinants---the fermionic equivalent of a product state.
The states $\ket{\psi_A}$ and $\ket{\psi_B}$ are reminiscent of genuinely entangled GHZ and W states for three qubits (cf.~\autoref{three qubits example}).
This remarkable correspondence between the entanglement classification of the Borland--Dennis setup and the classification of the three-qubit system can be explained precisely in a group-theoretical way \cite{levayvrana08}.

We now describe the corresponding entanglement polytopes (\autoref{borland dennis polytopes}):
The polytope of the first class is given by the solution of the quantum marginal problem, \eqref{borland dennis}, since the class is dense in the set of pure states.
The polytope of the second class is obtained by replacing the vertex $(0.5,0.5,0.5)$ (the origin) with $(2/3,2/3,2/3)$ (the local eigenvalues of the state $\ket{\Psi_B}$). This can be established geometrically by generalizing the argument used in the proof of \autoref{boson lemma}, cf.~\cite{kirwan84}.
The entanglement polytope of the third class is given by a line segment, and the polytope of the class of Slater determinants is given by single point. % (both assertions are easy to check by hand).
Thus the correspondence between the Borland--Dennis system and the three-qubit system is also manifest on the level of entanglement polytopes: the fermionic entanglement polytopes appear as the ``anti-symmetrization'' of the three-qubit polytopes (cf.~\autoref{three qubit polytopes}).

%%%% FIGURES

\clearpage

\begin{figure}
  \includegraphics[height=3.4cm]{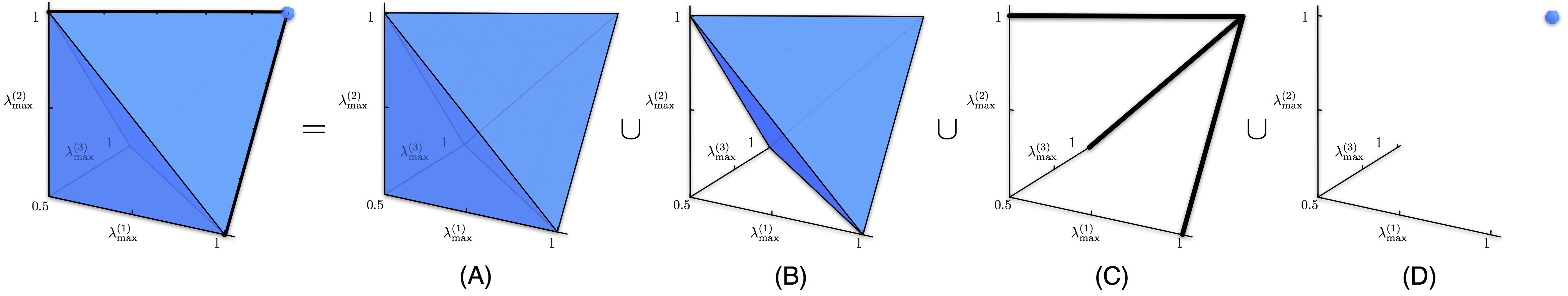}
  \caption{Entanglement polytopes for three qubits:
  (A) GHZ polytope (entire polytope, i.e., upper and lower pyramid),
  (B) W polytope (upper pyramid),
  (C) three polytopes corresponding to EPR pairs shared between any two of the three parties (three solid edges in the interior),
  (D) polytope of the unentangled states (interior vertex).}
  \label{three qubit polytopes}
\end{figure}

\clearpage

\begin{figure}
  \includegraphics[height=4cm]{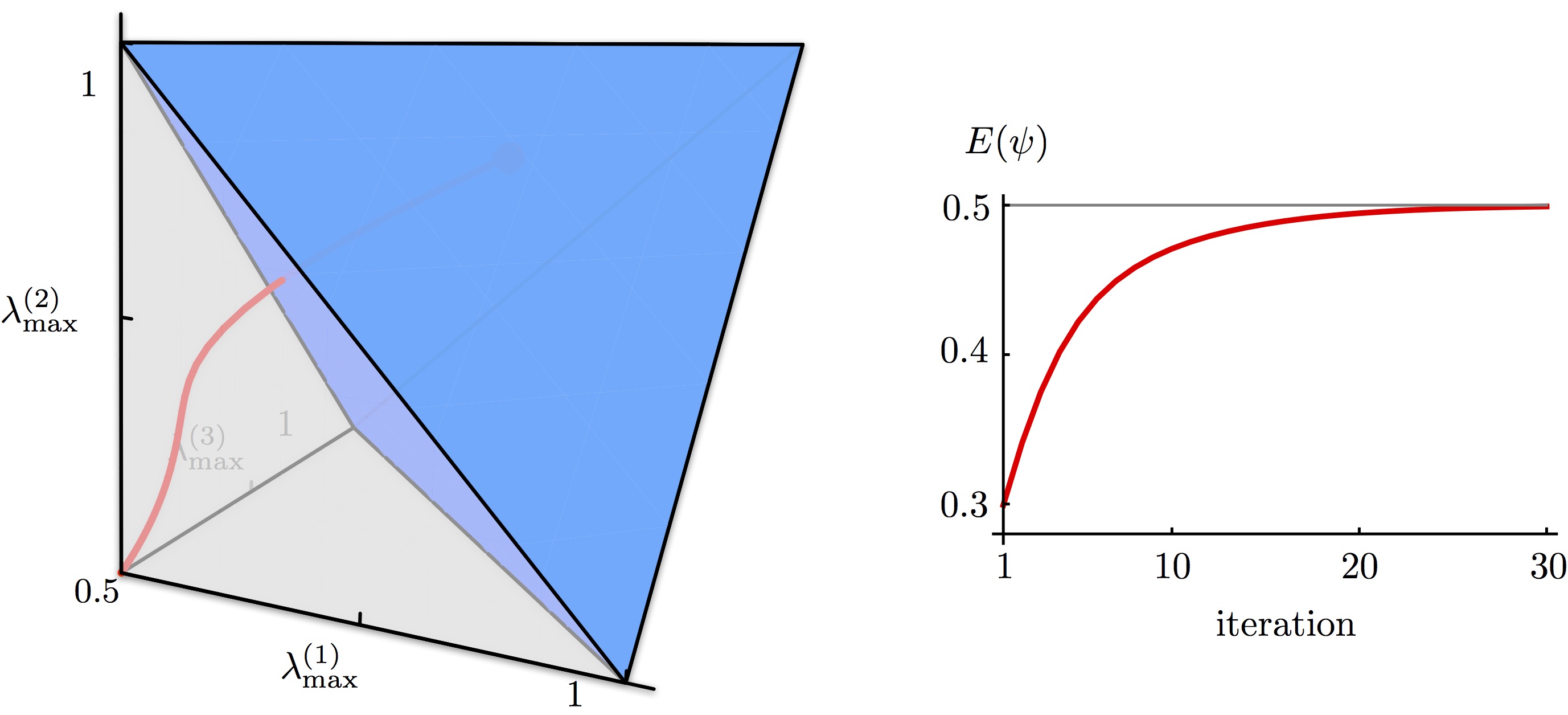}
  \caption{Illustration of entanglement distillation for the quantum state $\rho = \ketbra \psi \psi$ discussed in \autoref{three qubits example}: by following the gradient flow (curved line), the origin $O=(0.5,0.5,0.5)$ is reached asymptotically. The corresponding quantum state has maximal linear entropy of entanglement.}
  \label{three qubit distillation}
\end{figure}

\clearpage

\begin{figure}
  \includegraphics[width=0.9\textwidth]{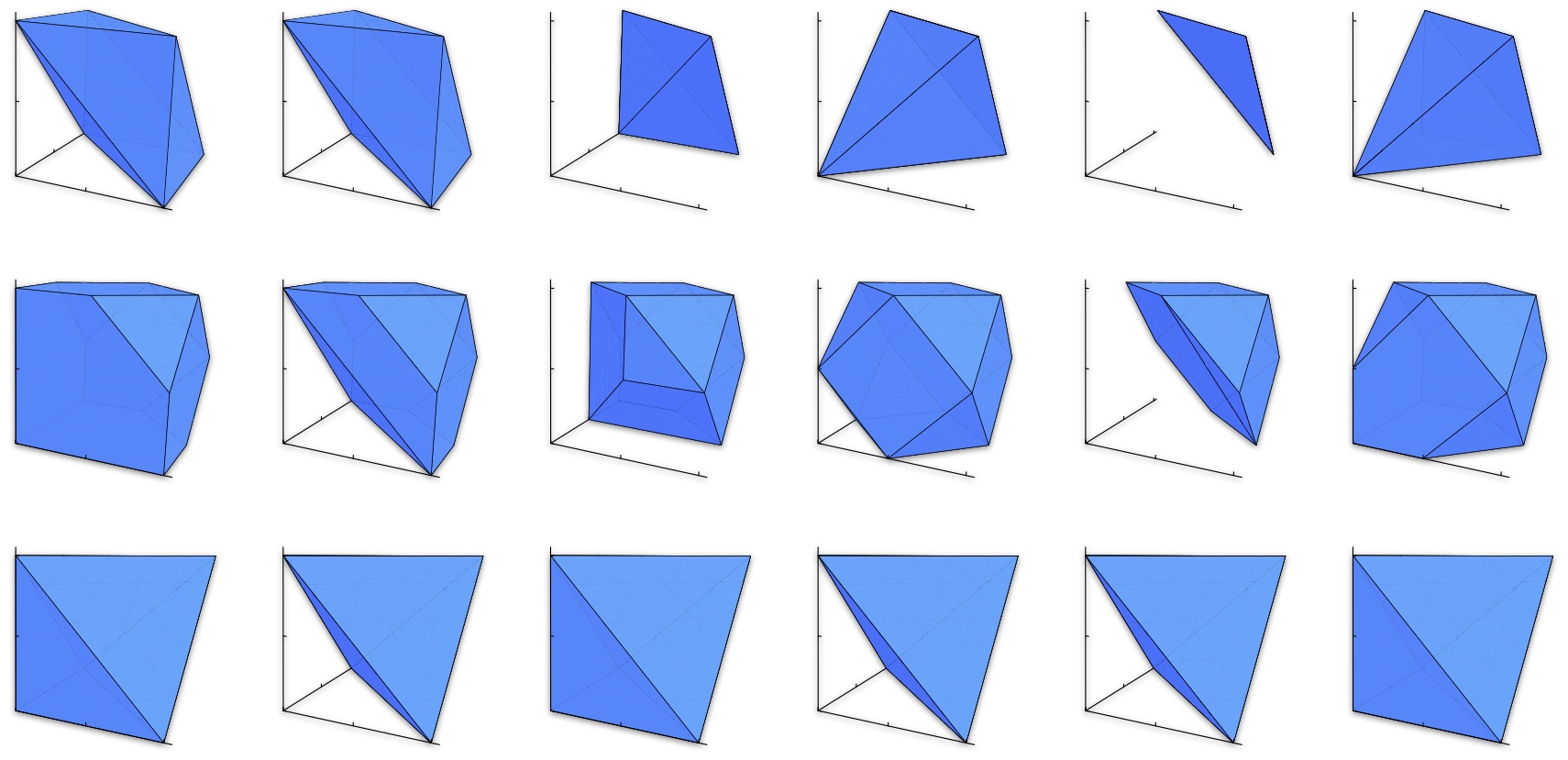}
  \caption{
  Three-dimensional cuts through the \emph{non-trivial} full-dimensional
  entanglement polytopes for four qubits. The columns correspond to the first
  six polytopes in \autoref{four qubits table}.
% (the 7th being the full polytope).
  The last coordinate is fixed in every row---to values
  $0.5$, $0.75$, and $1$, respectively, with the remaining three variables shown.
  \newline
  From the first row, one can see that only two of these entanglement polytopes
  (No.\ 4 and 6) include the maximally mixed vertex $(0.5, 0.5, 0.5,
  0.5)$. These reach the maximal value for the entropy of entanglement.
  The four polytopes missing the maximally mixed vertex cannot be
  distinguished from product states using polynomial entanglement
  monotones (c.f.~\autoref{entropy of entanglement}), showing that
  our approach is stronger in that regard.
  \newline
  The final row exhibits the behavior of the class when the fourth particle
  is projected onto a generic pure state. We then recover a three-qubit
  entanglement polytope on the remaining coordinates. Polytopes 1, 3,
  and 6 give the full three-qubit polytope implying
  that a GHZ-type vector is generated.
  Polytopes 2, 4, and 5 collapse to the upper pyramid of Fig.~1.  Hence these entanglement classes produce
  a $W$-type state when the first particle is projected onto a pure
  state. It follows that the mixed 3-tangle (and any other convex-roof
  extension of a polynomial monotone) vanishes on the mixed state
  generated by tracing out the last particle of states in these
  classes.
  This observation allows us to graphically recover some
  properties calculated algebraically in
  \cite{verstraetedehaenedemooretal02}, such as the vanishing
  3-tangle for the class $L_{{ab}_3}, a=b=0$ (which corresponds to
  polytope 5).
  }
  \label{fig:four qubits}
\end{figure}

\clearpage

\begin{figure}
  \includegraphics[height=3.6cm]{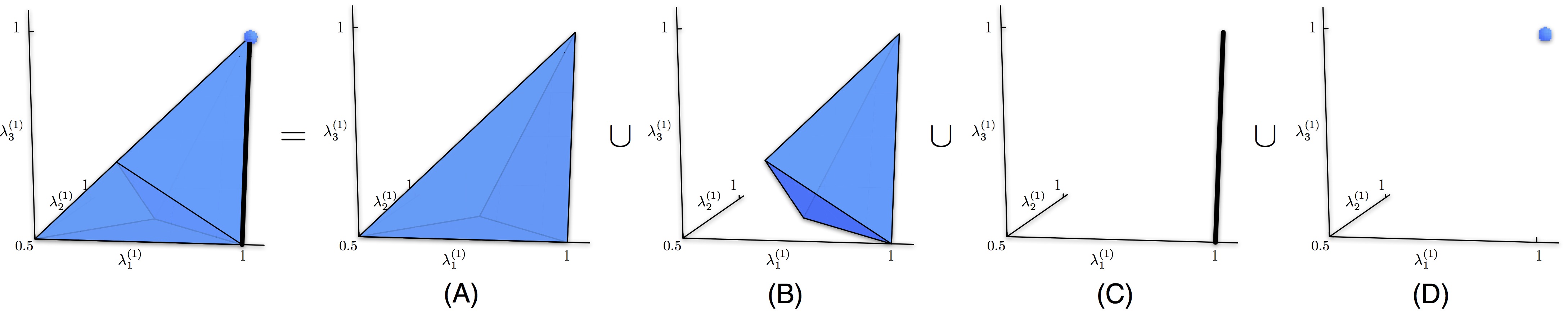}
  \caption{Entanglement polytopes for three fermions with local rank six:
  (A) class of $\ket{\psi_A}$,
  (B) class of $\ket{\psi_B}$,
  (C) class of fermionic ``biseparable'' states,
  (D) class of Slater determinants.}
  \label{borland dennis polytopes}
\end{figure}

% TABLES

\clearpage

\begin{table}
  \begin{tabular}{c|cc|cccccc}
    Covariant & Degree & Weight & $\ket\GHZ$ & $\ket W$ & $\ket{B1}$ & $\ket{B2}$ & $\ket{B3}$ & $\ket\SEP$ \\
    \hline
    $f$ & $1$ & $(1,1,1)$ & $\times$ & $\times$ & $\times$ & $\times$ & $\times$ & $\times$ \\
    $H_x$ & $2$ & $(2,1,1)$ & $\times$ & $\times$ & $\times$ & 0 & 0 & 0 \\
    $H_y$ & $2$ & $(1,2,1)$ & $\times$ & $\times$ & 0 & $\times$ & 0 & 0 \\
    $H_z$ & $2$ & $(1,1,2)$ & $\times$ & $\times$ & 0 & 0 & $\times$ & 0 \\
    $T$ & $3$ & $(2,2,2)$ & $\times$ & $\times$ & 0 & 0 & 0 & 0 \\
    $\Delta$ & $4$ & $(2,2,2)$ & $\times$ & 0 & 0 & 0 & 0 & 0 \\
  \end{tabular}
  \caption{Three-qubit covariants---labeled as in \cite[p.~33]{luque07}---with their degree and weight (see \autoref{examples} for conventions),
  and their vanishing behavior on the quantum states representing the six entanglement classes ($\times$ denotes non-vanishing).
  The invariant $\Delta$ is Cayley's hyperdeterminant, which is closely related to the 3-tangle \cite{coffmankunduwootters00,miyake03}.}
  \label{three qubit covariants}
\end{table}

\clearpage

\begin{table}
  \begin{tabular}{l|c}
    Representative & Number of embeddings \\
    \hline
    $\ket\GHZ \otimes \ket\uparrow$ & 4 \\
    $\ket W \otimes \ket\uparrow$ & 4 \\
    $\ket\EPR \otimes \ket\EPR$ & 3 \\
    $\ket\EPR \otimes \ket\uparrow \otimes \ket\uparrow$ & 6 \\
    $\ket\uparrow \otimes \ket\uparrow \otimes \ket\uparrow \otimes \ket\uparrow$ \qquad & 1
  \end{tabular}
  \caption{
    Bi-separable classes of four-qubit states. Each class corresponds to an
    embedding of a three-qubit entanglement polytope from
    \autoref{three qubits example} into the
    four-qubit polytope $\Delta$ in the obvious way. The right column
    gives the number of permutations of qubits that lead to distinct
    embeddings.
  }
  \label{four qubit lower dimensional}
\end{table}

\clearpage

\begin{table}
  \begin{tabular}{c|cl|cc|c|c|l}
    No & Family \cite{verstraetedehaenedemooretal02}~ & Parameters
    & Vertices~ & Facets &
    Perms & $E(\Delta_{\mathcal C})$ & Vertices compared to full
    polytope $\Delta$ \\
    \hline
    1 & $L_{0_{7\oplus \bar{1}}}$ & & 12 & 13 & 4 & 0.482143
    & $(0.5, 0.5, 0.5, 0.5)$ replaced by $(0.75, 0.5, 0.5, 0.5)$\\
    2 & $L_{0_{5\oplus 3}}$ & & 10 & 13 & 4 & 0.458333
    & $(0.5, 0.5, 0.5, 0.5), (1, 0.5, 0.5, 0.5)$ missing \\
    3 & $L_{a_4}$ & $a=0$ & 9 & 14 & 6 & 0.45
    & $(0.5, 0.5, 0.5, 0.5), (0.5, 1, 0.5, 0.5), (0.5, 0.5, 0.5, 1)$ missing  \\
    % 4 & $L_{\text{ab}_3}$ & $b=-a\neq 0$ & 8 & 16 & 1 & 0.5
    % & all three-partite entangled vertices missing \\
    % 4 & $L_{\text{abc}_2}$ & $a=c=0,b\neq 0$ & 8 & 16 & 1 & 0.5   <--   in my thesis
    % & all three-partite entangled vertices missing \\
    4 & $L_{\text{abc}_2}$ & $c = a = -b \neq 0$ & 8 & 16 & 1 & 0.5
    & all three-partite entangled vertices missing \\
    5 & $L_{\text{ab}_3}$ & $b=a=0$ & 7 & 9 & 1 & 0.375
    & all three- and four-partite entangled vertices missing \\
    6 & $L_{a_2b_2}$ &  $b=-a\neq0$ & 10 &14 & 6 & 0.5
    & $(0.5,1,0.5,0.5), (0.5,0.5,0.5,1)$ missing \\
    7 & $G_{abcd}$ & generic & 12 & 12 & 1 & 0.5 & n/a --- this is the full polytope
  \end{tabular}
  \caption{Entanglement polytopes for genuinely four-partite entangled states.
    The second column specifies one choice of family and parameters (according to the
        classification of \cite{verstraetedehaenedemooretal02}, as corrected in \cite{chterentaldokovic07}) which gives rise to the
    polytope shown (there might be further choices, partly because the parametrization
    of \cite{verstraetedehaenedemooretal02} is not always unique).
    ``Perms'' is the number of distinct polytopes one obtains when permuting the qubits. As
    defined in the main text, $E(\Delta_{\mathcal C})$ is
    the maximal linear entropy of entanglement in the class. The
    right-most column is a recipe for obtaining the given polytope out of
    the full one $\Delta$.}
  \label{four qubits table}
\end{table}

\end{document}